\documentclass{article}[11pt]
\usepackage[utf8]{inputenc}
\usepackage{graphicx}
\usepackage{amsmath,amsfonts,amssymb}
\usepackage{amsthm}
\usepackage{hyperref}
\usepackage{enumerate}
\usepackage{mathtools}
\usepackage{caption}
\usepackage{subcaption}
\usepackage{enumitem} 
\usepackage[margin=1in]{geometry}
\usepackage{cite}
\usepackage{authblk}
\usepackage{float}
\usepackage[parfill]{parskip}

\newcommand*{\bbZ}{\mathbb{Z}}
\newcommand*{\bbC}{\mathbb{C}}
\newcommand*{\bbD}{\mathbb{D}}

\newcommand*{\bbH}{\mathbb{H}}

\newcommand*{\cE}{\mathcal{E}}

\newcommand*{\cM}{\mathcal{M}}

\newcommand*{\cO}{\mathcal{O}}

\newcommand*{\cT}{\mathcal{T}}

\newcommand{\ket}[1]{|#1\rangle}

\newtheorem{theorem}{Theorem}

\newtheorem{definition}[theorem]{Definition}

\numberwithin{lemma}{section}

\DeclareMathOperator{\tr}{Tr}

\title{Building Bulk Geometry from the Tensor Radon Transform}

\begin{document}
\author[1]{ChunJun Cao\thanks{ccj991@gmail.com}}
\affil[1]{Joint Center for Quantum Information and Computer Science, University of Maryland, College Park, MD, 20742, USA}
\author[2]{Xiao-Liang Qi\thanks{xlqi@stanford.edu}}
\affil[2]{Stanford Institute for Theoretical Physics, Stanford University, Stanford, CA, 94305, USA}
\author[1,3]{Brian Swingle\thanks{bswingle@umd.edu}}
\affil[3]{Condensed Matter Theory Center, University of Maryland, College Park, MD, 20742, USA}
\author[4]{Eugene Tang\thanks{eugene.tang@caltech.edu}}
\affil[4]{Institute for Quantum Information and Matter, California Institute of Technology, Pasadena, CA, 91125, USA}
\maketitle

\begin{abstract}
    Using the tensor Radon transform and related numerical methods, we study how bulk geometries can be explicitly reconstructed from boundary entanglement entropies in the specific case of $\mathrm{AdS}_3/\mathrm{CFT}_2$. We find that, given the boundary entanglement entropies of a $2$d CFT, this framework provides a quantitative measure that detects whether the bulk dual is geometric in the perturbative (near AdS) limit. In the case where a well-defined bulk geometry exists, we explicitly reconstruct the unique bulk metric tensor once a gauge choice is made. We then examine the emergent bulk geometries for static and dynamical scenarios in holography and in many-body systems. Apart from the physics results, our work demonstrates that numerical methods are feasible and effective in the study of bulk reconstruction in AdS/CFT.
    \end{abstract}

\section{Introduction}

Recent progress \cite{VanRaamsdonk:2009ar,Ryu:2006bv,Ryu:2006ef,Swingle:2009bg,Maldacena:2013xja,qi2013exact,Cao:2016mst} in quantum gravity has shown that spacetime geometry can emerge from quantum entanglement. This emergence provides appealing explanations for many intuitive properties of the physical world, including the existence of gravity\cite{Lashkari:2013koa,Faulkner:2013ica,Swingle:2014uza,Faulkner:2017tkh}, conditions on the allowed distribution of energy and matter \cite{Lashkari:2014kda,Lashkari:2016idm}, and the unitarity of black hole dynamics \cite{Almheiri:2019psf,Penington:2019npb}. Most of these developments have taken place in the context of the Anti-de Sitter/Conformal Field Theory ($\mathrm{AdS}/\mathrm{CFT}$) correspondence\cite{Maldacena:1997re,Aharony:1999ti}, although some proposals relating geometry and entanglement also apply to flat or de~Sitter geometries\cite{Maldacena:2013xja,Cao:2016mst,Cao:2017hrv,Czech:2015kbp,Bao:2017qmt}. As an example of the holographic principle\cite{tHooft:1993dmi,Susskind:1994vu}, $\mathrm{AdS}/\mathrm{CFT}$ describes a duality between a $(d+1)$-dimensional bulk theory with the presence of gravity in asymptotic AdS spacetime, and a $d$-dimensional conformal field theory (CFT) without gravity on the boundary. For the purposes of studying quantum gravity, the duality is especially powerful because the CFT is an object whose basic rules we understand. However, much remains to be understood about the relationships between the two sides of the duality.

One such challenge in the AdS/CFT duality is to understand how the boundary degrees of freedom without gravity can reorganize themselves into a higher dimensional bulk configuration with gravity. 
This is called the problem of bulk reconstruction, and this paper reports two results on this topic. First, we describe a perturbative procedure to reconstruct the bulk geometry given an appropriate set of boundary entanglement data. Second, we show that this reconstruction procedure can detect whether the putative bulk dual is semi-classical in the sense of having only weak fluctuations about an average value.

Our first result builds on a number of works that study bulk metric reconstruction using geodesic lengths \cite{Porrati:2003na,Roy:2018ehv} or entanglement entropy\cite{Czech:2014ppa}.\footnote{ More generally, but less explicitly, there are also proposals for bulk reconstruction using tensor networks\cite{Swingle:2009bg,Pastawski:2015qua,Hayden:2016cfa,Evenbly:2017hyg,you2018machine}, modular Hamiltonians\cite{Kabat:2018smf}, and light-cone cuts \cite{Engelhardt:2016crc,Engelhardt:2016wgb}. These rely on different boundary data and methodologies which we will not discuss here.} Early efforts in this area often utilized bulk symmetries to simplify the problem of recovering the bulk metric from minimal geodesic data~\cite{Hammersley:2006cp,Bilson:2008ab,Bilson:2010ff}. Later work on differential entropy and hole-ography~\cite{Balasubramanian:2013lsa,Czech:2014ppa} furthered our understanding using information theoretic quantities and suggested a method to recover the bulk, although no explicit reconstruction formula was given for generic cases. More recently, it was shown that in certain cases, static geometries\cite{Bale2017}, or even the full dynamical metric\cite{Bao:2019bib}, can be fixed non-perturbatively by boundary entanglement data without any prior knowledge of the bulk symmetries. However, barring a few well-known examples with symmetries, there does not exist a reconstruction procedure that directly and explicitly converts entropy data into bulk metrics. Our work addresses this missing element.

Our second result arises from the basic issue that we have a far from complete understanding of what kind of boundary states correspond to semi-classical bulk geometries. Some necessary conditions are expected from holographic entropy inequalities\cite{Bao:2015bfa} and from consistency relations for any putative metric reconstruction procedure\footnote{ One typically proceeds by assuming the boundary state has a bulk geometry and running the reconstruction procedure. If the boundary state actually does not have a bulk geometry, then the reconstruction procedure will lead to internal inconsistencies.}. In general, we do not expect all quantum states from the boundary CFT to correspond to well-defined semi-classical geometries in the bulk. On the contrary, we expect an abundance of non-geometrical states obtained, for instance, by superposing states with macroscopically distinct dual geometries such as the AdS vacuum and an AdS black-hole geometry. Therefore, to better understand holographic duality, we must also address the necessary and sufficient conditions under which a bulk geometry can emerge from the boundary state. Our work also addresses this open question.


In this paper, we study the above issues in the context of  $\mathrm{AdS}_3/\mathrm{CFT}_2$ using an approach based on the tensor Radon transform. The method is for metric solutions that are close to $\mathrm{AdS}_3$, hence it is restricted at present to perturbative problems. To linear order in the perturbation, each quantum state on the boundary corresponds to a constant time slice in the bulk, so we provide a numerical reconstruction algorithm that takes the entanglement entropies of intervals in the boundary state as input, and outputs the best-fit bulk metric tensor of the spatial slice in the linearized regime. This solution is unique up to gauge transformations. The algorithm also provides a quantitative indicator of whether the boundary data admits a bulk geometric description near $\mathrm{AdS}_3$. This is measured by the quality of the fit, which intuitively quantifies how far the boundary entanglement data of a given state is from being geometric. A poor quality of the fit indicates the boundary data lack consistency with a semi-classical geometry.

As a proof of principle, we explore several reconstructions numerically in holography and with a $1$d free fermion CFT. We find that free fermion ground states in the presence of disorder and a mass deformation do not correspond to a well-defined bulk geometries. Likewise, mixtures of states where each is dual to a distinct classical geometry can also fail to have a well-defined bulk geometry. In addition to these static examples, the method is applied to several dynamical scenarios including global and local quenches in the free fermion model and entanglement dynamics in a toy model scrambling system. In the case where the dynamics is scrambling, we find that the bulk description is qualitatively consistent with an in-falling spherical shell of bulk matter experiencing gravitational attraction. These results further demonstrate that it is both feasible and interesting to study AdS/CFT using tensor Radon transform techniques coupled with modest (laptop-scale) computational resources.

The remainder of the paper is organized as follows. In section \ref{sec:boundaryrig}, we briefly review the basic assumptions, especially how entanglement entropy can be tied to metric tensors via the tensor Radon transform. We give a general review on the tensor Radon transform in Appendix \ref{app:TRT}. Its adaptation to a hyperbolic geometry and the gauge fixing procedure are found in Appendix \ref{app:TRTpoincare} and Appendix \ref{app:gaugefix}, respectively. In section \ref{sec:numrec}, we introduce the numerical reconstruction procedure, which we elaborate in detail in Appendix \ref{app:numrecon}. In section \ref{sec:recongeo}, we apply the reconstruction algorithm to static and dynamical boundary entanglement data. In section \ref{sec:geodetect}, we discuss these reconstructions and how geometrical and non-geometrical can be distinguished using the relative reconstruction error. Finally we conclude with some remarks and directions for future work in section \ref{sec:discussion}.

\section{Boundary Rigidity and Bulk Metric Reconstruction}
\label{sec:boundaryrig}

We begin with the vacuum state $|0\rangle_{\mathrm{CFT}}$ of a hologaphic CFT. Because this state has conformal symmetry at all scales, it must be dual to empty AdS \cite{Maldacena:1997re}. When the bulk theory is Einstein gravity coupled to matter, the bulk geometry controls the leading entanglement structure of the CFT state via the Ryu-Takayanagi (RT) formula \cite{Ryu:2006bv,Ryu:2006ef}. Given a boundary region $A$, the RT formula computes the von Neumann entropy $S(A)$ in terms of a minimal area surface,
\begin{equation}
    S(A) = \frac{1}{4G}\min_{\gamma_A} \mathrm{Area}[\gamma_A],
\end{equation}
where $\mathrm{Area}[\gamma]$ is the area of $\gamma_A$, $G$ is Newton's constant, and the minimum runs over bulk surfaces $\gamma_A$ that are homologous to $A$. In the time symmetric case where RT applies, all of these surfaces can be taken to lie in a time-symmetric spacelike surface $\Sigma$. 
In other words, the Ryu-Takayanagi formula says that the von Neumann entropy of a state on a boundary subregion is given by area of the minimal area bulk surface that subtends the region. 

If we have access to a boundary state, in the sense of knowing its von Neumann entropies on all (connected) subregions, then the RT formula translates these entropic quantities into a set of boundary anchored minimal surface areas. Because these minimal surfaces all lie on the spatial slice $\Sigma$, recovering the bulk geometry from entanglement reduces to a pure geometry problem where we try to find the interior metric $g_{ij}$ of a Riemannian manifold $\cM$ while knowing only the areas of minimal surfaces that are anchored to its boundary $\partial \cM$. This is precisely the statement of the boundary rigidity problem\cite{UhlSteAnnounce}, which is well-studied in the field of integral geometry\cite{sharafutdinov1994integral}. 

For this work, we focus exclusively on the case of $\mathrm{AdS_3/CFT_2}$, where minimal surfaces are simply geodesics and spatial slices are 2d Riemannian manifolds. For a class of 2d Riemanian manifolds (called \emph{simple} manifolds), it is known that the lengths of all boundary-anchored geodesics indeed fixes the bulk metric uniquely up to gauge equivalence~\cite{PestovUhlmann2003}. Because the single interval von Neumann entropies in the ground state of a 2d CFT are universally determined by the central charge~\cite{Calabrese:2009qy}, the RT formula combined with boundary rigidity completely fixes the bulk geometry to be that of hyperbolic space\cite{BCG1995,Croke1990}. Of course, this is precisely the induced geometry on a time-symmetric slice of AdS, as it had to be based on the grounds of symmetry\footnote{Since the lengths of geodesics on an asymptotically AdS spacetime are divergent when extended to the boundary, it is understood that the geodesics are actually regularized to be anchored on a cutoff surface that is some finite distance away from an arbitrarily chosen coordinate in the bulk. For the rest of this work, we will always impose a UV cutoff for the CFT, or equivalently, an IR cut off in the bulk geometry so that the geodesic lengths stay finite.}.

It is natural to ask whether we can exploit the power of the RT formula and the results from boundary rigidity theory to reconstruct dual geometries from boundary states other than the vacuum. For instance, given a generic non-vacuum state $\ket{\psi}_{\mathrm{CFT}}$, can we apply the same principles to reconstruct the metric tensor for the bulk geometry from the set of boundary-anchored geodesic lengths? There are several obstacles that prevent us from recovering the bulk metric tensor exactly using the above methods, even if $\ket{\psi}_{\mathrm{CFT}}$ has a well-defined dual geometry. First, since it may not be possible for minimal surfaces (or extremal surfaces in the dynamical case\cite{Hubeny:2007xt}) to foliate entire space(time) manifold, thus we can only reconstruct regions where there is at least a local foliation with minimal or extremal surfaces. Second, even for the regions whose geometries are fixed by the entanglement data\cite{Bao:2019bib}, there is no explicit reconstruction formula for the general boundary rigidity problem.

Although such problems are difficult to solve in general, it is typically easier to reconstruct the difference between the dual geometry and a known reference, or background, geometry. This is known as the linearized boundary rigidity  problem\cite{sharafutdinov1994integral,UhlSteAnnounce}. In this work, instead of a direct reconstruction of the of the dual geometry for $\ket{\psi}_{\mathrm{CFT}}$, we reconstruct the differences in the entanglement patterns as linearized metric perturbations around the AdS background.

To do so, we first fix the background geometry to be vacuum $\mathrm{AdS}_3$. Working with a given constant-time slice, let us suppose that $\ket{\psi}_{\mathrm{CFT}}$ has a slightly different entanglement structure compared to $\ket{0}_{\mathrm{CFT}}$, and is dual to a bulk geometry with a metric that is close, but not equal, to that of pure hyperbolic space on our time-slice. Then by the RT formula, the boundary-anchored geodesic lengths now differ slightly from those of pure hyperbolic space. For a given boundary subregion $A$, the change in the geodesic length anchored at $A$ is related to the vacuum subtracted entropy of the state by
\begin{align}
    \Delta L(A) = L_\psi(A)-L_{0}(A) = \frac{S_{\psi}(A)-S_0(A)}{4G},
\end{align}
where $L_{\psi}(A)$ and $L_0(A)$ denotes the lengths of geodesics anchored at the end points of $A$ for states $\ket{\psi}_{\mathrm{CFT}}$ and $\ket{0}_{\mathrm{CFT}}$, respectively. 

This change in geodesic length corresponds to a change in the bulk metric
\begin{align}
    g_{ij}^{(0)}\mapsto g_{ij}=g_{ij}^{(0)}+ h_{ij},
\end{align} 
where $g^{(0)}_{ij}$ is the pure hyperbolic metric and $h_{ij}$ is the perturbation. For the linearized problem, the goal is to find $h_{ij}$ to leading order, with $h_{ij}$ viewed as a rank two symmetric tensor field on the hyperbolic background. Note that geodesics of the background metric remain geodesics of the perturbed metric to first order in $h$ since geodesics satisfy an extremality condition. Changes in entanglement due to variations in the minimal surface itself are of order $h^2$ \cite{Lashkari:2013koa}. 
The leading order change in geodesic length can then be written as
\begin{align}
    \Delta L(A) &= L_\psi(A)-L_0(A)\nonumber\\
    &= \int_{\gamma_A} \sqrt{(g^{(0)}_{ij}+h_{ij})\dot{\gamma}_A^i\dot{\gamma}_A^j}\ ds-\int_{\gamma_A}\sqrt{g^{(0)}_{ij}\dot{\gamma}_A^i\dot{\gamma}_A^j}\ ds\nonumber \\
    &= \frac{1}{2} \int_{\gamma_A}h_{ij} \dot{\gamma}_A^i\dot{\gamma}_A^j\ ds + \mathcal{O}(h^2),
    \label{eqn:leadingorderd}
\end{align}
where $\gamma_A$ is the geodesic of the hyperbolic background anchored at the end points of $A$ and $\dot{\gamma}_A^i$ denotes unit tangent vectors along $\gamma_A$. The tangent vectors are normalized such that
\begin{align}
g^{(0)}_{ij}\dot{\gamma}_A^i\dot{\gamma}_A^j =1.
\end{align} 

For simplicity, we will use $L(\gamma_A)$ and $L(A)$ interchangeably, often dropping the explicit dependence on $A$ and simply writing $L(\gamma)$ when there is no confusion, with the understanding that $\gamma$ is a boundary-anchored geodesic. 

To make concrete progress in this work, we simplify the full problem by taking the linearized bulk result for changes in geodesic length to be equal to the full change in boundary entanglement entropy,
\begin{equation}
     \Delta L(\gamma) \approx \frac 1 2 \int_{\gamma} \dot{\gamma}^i\dot{\gamma}^jh_{ij}\ ds.
\end{equation}
This approximation enables calculations; going beyond it may be technically non-trivial, but it is likely not a fundamental obstacle.

With this simplification, the length perturbation becomes the integrated longitudinal projection of the metric perturbation along a geodesic $\gamma$. This is precisely the \emph{tensor Radon transform} $R_2[h]$ of the metric perturbation $h_{ij}$ \cite{Czech:2016tqr}. Given a symmetric $2$-tensor field $h_{ij}$, the tensor Radon transform $R_2[h]$ defines a map from the space of boundary geodesics to the complex numbers given by
\begin{equation}
 R_2[h_{ij}](\gamma_A)\equiv\int_{\gamma_A} h_{ij}\dot{\gamma}_A^i(s)\dot{\gamma}_A^j(s)\ ds,
    \label{eqn:linearRdT}
\end{equation}
where $\gamma_A$ is the background geodesic anchored at the boundary of $A$.

As an aside, we note that there exist several related notions of Radon transform. A standard Radon transform on a Riemannian manifold is defined by integrating some quantity on a minimal co-dimension one surface, whereas an X-ray Radon transform is defined similarly, but for a dimension one surface, i.e., a geodesic. For two spatial dimensions, as is the case we consider here, the two definitions coincide. For more details on the Radon transform, see Appendix \ref{app:TRT}. 

Formally, the the bulk metric deformation can be recovered by inverting the tensor Radon transform. Schematically, we can write
\begin{equation}
    R_2^{-1}[2\Delta L] = R_2^{-1}[R_2[h_{ij}]]= h_{ij},
\end{equation}
where $R_2^{-1}$ denotes some (not yet properly defined) inverse Radon transform, and where $\Delta L$ denotes the collection of boundary anchored geodesic length deviations. 

Throughout, we work exclusively in the perturbative regime, to leading order in $h$, which allows us to relate geodesic length deformations to the Radon transform through~\eqref{eqn:leadingorderd}. It also ensures that the resulting geometric solution, when it exists, is uniquely determined by the boundary entropy data~\cite{Croke1990,PestovUhlmann2003}.\footnote{In general, we are not guaranteed a unique solution for the bulk metric, even if one exists, because the assumption that the manifold is simple breaks down for sufficiently large deviations from a constant curvature background~\cite{Porrati:2003na}. Working in the perturbative limit ensures that the Radon transform remains well-defined. See Appendix~\ref{app:TRT}.} We wish to comment here that despite restrictions to the perturbative regime, $h_{ij}$ can still capture highly non-trivial physics. Indeed, standard calculations of gravitational waves and the dynamics of typical stars, planets, and galaxies are all done in the weak-field regime.

Before we can proceed with the inversion process, we must give meaning to the inverse Radon transform. For this purpose, it is important to note that the Radon transform has a non-trivial kernel: given any vector field $\xi$ on $M$ such that  $\xi|_{\partial M}=0$, we necessarily have 
\begin{equation}
    R_2[\nabla_i\xi_j + \nabla_j \xi_i] = 0. 
\end{equation}
Physically, any Radon transform of a pure gauge deformation that reduces to the identity at the boundary is zero. The presence of this kernel is natural because the transform relates geodesic lengths (which are gauge invariant) to metric tensors (which are not), so the Radon transform can only be injective up to gauge.

In the presence of such a kernel, we must fix a gauge prescription in order to recover a metric tensor uniquely. We will use a prescription which we call the \emph{holomorphic gauge}~\cite{Monard2015}. In a crude sense, the holomorphic gauge preferentially reconstructs the trace part of the metric at the cost of diminishing non-zero contributions to the off-diagonal.\footnote{There exists other gauge fixing prescriptions as well. The most commonly considered prescription is known as the~\emph{solenoidal gauge}. See Appendix \ref{app:TRT}. We choose an alternative gauge prescription for various reasons of convenience. For more details on gauge fixing, see Appendix~\ref{app:gaugefix}.} This provides two independent gauge constraints in two spatial dimensions, which would allow us to proceed with the reconstruction. For a more detailed description of the gauge constraints and Radon transform on a hyperbolic background, see Appendix~\ref{app:TRTpoincare}. 

As a final comment, in the above discussion we have explicitly assumed that the boundary state corresponded to a well-defined bulk geometry. Below we will construct algorithmic machinery which can actually carry out the reconstruction in this case. However, we will also see that the algorithm can be applied to a more general class of states, with interesting results.

\section{Numerical Methods for Reconstruction}
\label{sec:numrec}

The inversion formula for the flat-space scalar Radon transform is a well-known classical result in integral geometry~\cite{sharafutdinov1994integral}. Explicit reconstruction formulas for scalar and vector Radon transforms are also available for constant negative curvature backgrounds~\cite{PestovUhlmann2003,Venka2010}. However, there are currently no explicit reconstruction formulas available for higher rank tensors on curved backgrounds, although several results in the literature come close to a solution in various regimes~\cite{PestovUhlmann2003,UhlmannVasy,Monard2015,Monard2014,Venka2010,UhlmannVasy,BAL20051362}. In the absence of an exact analytic reconstruction formula, we instead draw inspiration from the general principles employed in seismology to study the Earth's interior~\cite{dahlen_tromp_1999} in developing our numerical method. 

In this section, we give a brief overview of the method. The full details of the discretization, gauge fixing, and solutions for the constrained least square problem can be found in Appendix~\ref{app:numrecon}.

\subsection{Discretization and Optimization Procedures}

The basic idea behind our numerical reconstruction is straightforward. We first discretize the bulk and boundary regions into a finite number of tiles. To each tile $\cT$ in the bulk, we associate a tensor $h_{ij}(\cT)$, and for each interval $A$ on the boundary, we associate a geodesic (of the background metric) $\gamma_{A}$ anchored on the endpoints of the interval. In two spatial dimensions, a rank-$2$ symmetric tensor has $3$ independent degrees of freedom. Each geodesic anchored at the end points of an interval generates a linear equation via the discretized version of the Radon transform (\ref{eqn:linearRdT}), defined by
\begin{align}
   2\Delta L(\gamma_{A}) =  R_2[h_{ij}] \approx \sum_{\cT}W^{ij}(\cT,\gamma_{A})h_{ij}(\cT),\label{eq:discrete_RT},
\end{align}
where the tensor $W^{ij}$ contains information about the direction of the tangent vectors $\dot{\gamma}_{A}$, as well as the arc length $\Delta s(\gamma_{A},\cT)$ of the geodesic segment that passes through each tile $\cT$. In equation~\eqref{eq:discrete_RT}, we sum over all bulk tiles $\cT$ and over repeated indices. Naturally, one has $W^{ij}(\cT, \gamma_{A})=0$ if the geodesic does not pass through a tile $\cT$. 

We can abbreviate equation \eqref{eq:discrete_RT} in matrix form as
\begin{align}
    \mathbf{b} = \mathbf{Wh},
\end{align}
where $\mathbf{W}$ and $\mathbf{h}$ are vectorized representations of $W^{ij}$ and $h_{ij}$, and where $\mathbf{b}$ denotes the corresponding geodesic length deformations. As a result, given a specific discretization, the discretized forward Radon transform can be written as a linear map $W: V_{B}\rightarrow V_{\gamma}$ from the space of tile-wise constant bulk tensor valued functions $V_B$ to the space of boundary anchored geodesic lengths $V_{\gamma}$. Both spaces are finite dimensional due to the discretization.

Since the forward Radon transform has a non-trivial kernel in the continuum limit, we must impose a gauge fixing condition to recover a unique solution. We give the full detail of the gauge fixing conditions and the accompanying partial differential equations in Appendix \ref{app:gaugefix}. To ensure the problem is well-posed, we set the holomorphic gauge constraints as discretized partial differential equations, which we formally write as
\begin{align}
    \mathbf{Ch} = \mathbf{0}.\label{eq:discrete_constraints}
\end{align}
Here, $\mathbf{C}$ denotes the constraint matrix representing the partial differential operator associated with the gauge constraint. Reconstruction of the metric perturbation then corresponds to finding the solution to the linear equations \eqref{eq:discrete_RT} above, subject to the linear constraints \eqref{eq:discrete_constraints}.

In practice, there does not always exist exact solutions $h_{ij}$ which satisfies the constrained system. This can be due to a variety of reasons, such as the presence of discretization errors, or if the boundary data is simply inconsistent with a geometric bulk, i.e., if the boundary entropy function fails to lie within the range of the forward Radon transform \cite{sharafutdinov1994integral,PestovUhlmann2003,Monard2015}. Instead of trying to look for an exact solution, it is more natural to look for the best-fit solution $\mathbf{h}_*$ which solves the constrained minimization problem
\begin{align}
\min_{\mathbf{h}}\|\mathbf{Wh}-\mathbf{b}\|, \nonumber\\
\quad \text{subject to}~\mathbf{Ch} = \mathbf{0}. \label{eq:lstsq0}
\end{align}
The objective function is linear and we are guaranteed a unique global minimum. Thus we will say that $\mathbf{h}_*$ is the optimal geometric solution corresponding to boundary data $\mathbf{b}$. We will also write $\mathbf{h}_*(\mathbf{b})$ when we need to denote the dependence of $\mathbf{h}_*$ on the initial boundary data.

Even with the existence of an optimal reconstruction $\mathbf{h}_*$, we do not expect generic boundary data to correspond to a geometric dual in general. A useful quantity is the relative error of reconstruction, which measures the tension between the best-fit solution and the actual data. We can consider various relative errors. If we know the exact bulk solution, say $\mathbf{h}_0$, then we can denote the bulk relative error as
\begin{align}
    \cE_{\mathrm{bulk}} = \frac{\|\mathbf{h}_*-\mathbf{h}_0\|}{\|\mathbf{h}_0\|},
\end{align}
where $\mathbf{h}_*$ is the bulk metric tensor reconstructed from the forward transform of $\mathbf{h}$. 

More commonly however, we do not have access to an a priori geometric state. Instead, we have a CFT state $\ket{\psi}_\mathrm{CFT}$ from which we can extract discretized boundary data $\mathbf{b}_\psi$. In this case, we can likewise consider the boundary relative error, defined by
\begin{align}
    \cE_{\mathrm{bdy}} = \frac{\|\mathbf{W}[\mathbf{h}_*(\mathbf{b}_\psi)] -  \mathbf{b}_\psi\|}{\mathbf{\|b_\psi\|}}.
    \label{eqn:reconerr}
\end{align}
The boundary relative error is simply the normalized distance from $\mathbf{b}_\psi$ to the subspace of boundary data vectors with geometric duals, which is the same quantity minimized by the constrained least squares problem~\eqref{eq:lstsq0}. The boundary relative error therefore serves to quantify the degree to which a state is geometric or non-geometric. We will discuss in greater detail in section~\ref{sec:geodetect} how the relative errors can be used to distinguish geometric data from non-geometric data on the boundary. 

To ensure the reliability of the reconstruction algorithm for the inverse tensor Radon transform, we also perform the numerical inversion of boundary data $\mathbf{b}$ whose bulk tensor field is known. We produce such boundary data by preparing various known bulk tensor valued functions $\mathbf{h}_0$ in the holomorphic gauge then generating its corresponding geodesic data $\mathbf{b_0}$ through a forward tensor Radon transform. Subsequently, we apply the numerical reconstruction to the geodesic data and compare the reconstructed $\mathbf{h_*}(\mathbf{b_0})$ to the original test function $\mathbf{h}_0$. We find remarkable agreement in our reconstructions. Absent rigorous analytic convergence guarantees, the successful benchmarking of the algorithm on known cases serve to provide confidence in the fidelity of the reconstruction. Details of this benchmarking process are given in Appendix \ref{app:benchmark}.

\section{Reconstructed geometries} 
\label{sec:recongeo}
The forward tensor Radon transform is generally neither injective nor surjective. Therefore, not all boundary data can be interpreted as the Radon transform of a bulk tensor field.  However, it may be possible to derive necessary and sufficient characterization of what boundary data corresponds to a bulk tensor field. Such a criterion is known as a \emph{range characterization} of the tensor Radon transform. Ranges characterizations of various transforms have been discussed extensively for scalar and vector cases on curved backgrounds and tensor cases on flat background~\cite{sharafutdinov1994integral,PestovUhlmann2003,Monard2015}. However, it remains an active topic of research for transforms on curved background for higher rank tensor fields.

Although we lack a rigorous analytic characterization for the tensor Radon transform, our numerical methods can still effectively capture the parts of the entanglement data that do not lie within the range of the tensor Radon transform. Since the Radon transform is linear, any contribution that cannot be fitted to a bulk tensor field in the global best-fit reconstruction effectively captures the non-geometric contribution for the discrete reconstruction. More specifically, the relative boundary error~\eqref{eqn:reconerr} serves as a indicator for the fraction of the boundary entropy data which does not lie within the range, i.e., the fraction which can be considered non-geometric.

In the upcoming sections, we reconstruct geometries from boundary entanglement entropies generated by holographic systems as well as those generated numerically from a $1$d free fermion system.  We then discuss how their relative boundary errors $\cE_{\rm bdy}$ can be used as a standard to distinguish states that have a well-defined classical dual geometry from the ones that do not. This provides a direct quantitative condition for whether a state is geometric.

For entanglement entropies generated by holographic systems, we expect a weakly coupled gravity dual, and therefore spacetime geometries that are generically classical at low energies. The same can not be expected for a generic many-body quantum system at criticality~\cite{Heemskerk:2009pn}, although they may still capture interesting features using a dual spacetime prescription. For instance, in a free fermion system, the conformal field theory has neither strong coupling nor a large central charge. Although it is unclear what the dual bulk description would be, it is generally expected that any geometric description must be one where gravity is strongly coupled, the leading order RT formula does not apply, and the spacetime is dominated by quantum effects. Therefore, the reconstruction is likely poor in the absence of a large number of symmetries. We will find that numerical evidence support these basic intuitions.

It is difficult to present certain reconstructions from dynamical systems in a paper. A list of the animated reconstructions for various dynamical processes are  \href{https://www.youtube.com/playlist?list=PLCjJ3kjqxOfw1aIa5c0X6KSpox1-AjM5b}{linked here}.\footnote{Link: https://www.youtube.com/playlist?list=PLCjJ3kjqxOfw1aIa5c0X6KSpox1-AjM5b}

\subsection{Holographic Reconstructions}
\label{subsec:holorecon}
As a benchmark for holographic geometry reconstructions, we first reconstruct the metric for the thermal AdS geometry in the linearized regime. The entropy data we use are generated by the minimal geodesic lengths in a BTZ geometry~\cite{Banados:1992wn} using the Ryu-Takayanagi formula. Since the data corresponds to a bulk geometry by construction, the reconstruction should show good agreement with said geometry at linear order. This includes the correct qualitative behaviour in the bulk, and a positive metric perturbation present deeper into the bulk. It must also afford relatively small reconstruction errors, which can be attributed to factors such as discretization errors, approximations made in linearization, and working on a fixed hyperbolic background despite the changes in bulk geometry.

\begin{figure}
    \centering
    \includegraphics[width=0.7\textwidth]{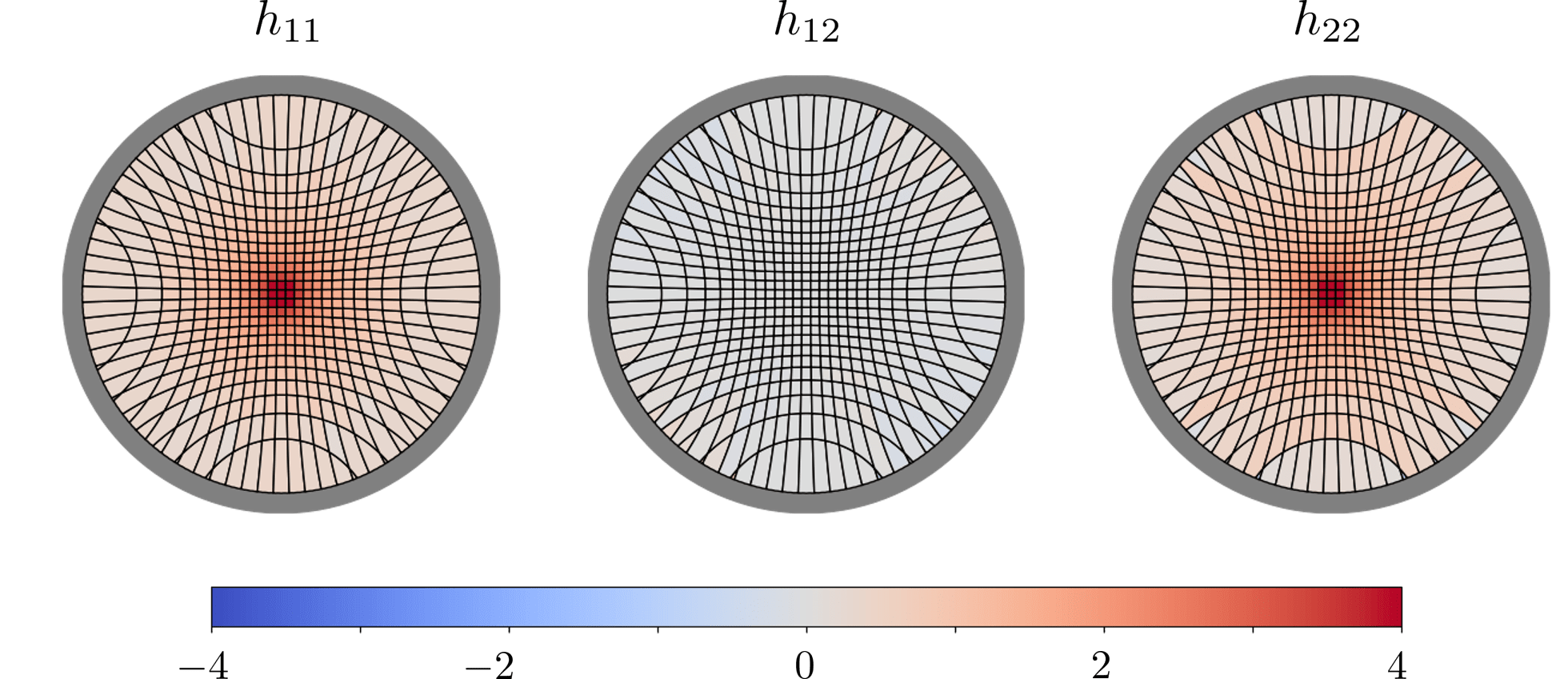}
    \caption{Thermal state reconstruction. Plot shows the individual tensor components $h_{11},h_{12},h_{22}$ of the metric perturbation from left to right, respectively. The boundary relative errors are $\cE_{\rm bdy}\lesssim 10^{-2}$ across a wide range of temperatures.}
    \label{fig:TAdSRecon}
\end{figure}

Indeed, as shown in Figure~\ref{fig:TAdSRecon}, we find good agreement with our expectations in the reconstruction. By probing deeper into the bulk, the geodesics for a thermal geometry travel through larger distances, resulting in a net positive change in the metric perturbation. Note that at linear order, we can not detect a change in bulk topology, which the entanglement data dual to a BTZ geometry should predict. This is because of the fixed background metric.\footnote{Suppose the background metric is updated using the new found metric perturbation. Then we should recover the fact that the corrected geodesics avoid the central region of the bulk. However, geodesic avoidance alone does not necessarily indicate the formation of a horizon. For instance, adding a single massive particle in AdS will lead to a back-reacted geometry where geodesics avoid the region near the inserted mass.} 

Motivated by entanglement dynamics in holography, we can also construct heuristics for the entanglement growth of a boundary theory. For instance, under a global quench, we expect qualitative growth of entanglement to be captured by
\begin{equation}
    S_A(t) \approx \min \{svt, s|A|\}\label{eq:vollaw}
\end{equation}
for $t\ge 0$, where entanglement for any region will grow linearly in time after the initial stage, until the entropy satisfies a volume law~\cite{Liu:2013qca,Couch:2019zni}. Here $s$ denotes the entropy density, $v$ the speed of entropy growth, and $t$ the time that has elapsed since the quench. The size of the boundary interval is denoted by $|A|$. The corresponding metric and curvature perturbations are shown in Figure~\ref{fig:thermalization}.

\begin{figure}
\centering
\begin{subfigure}{0.7\textwidth}
  \includegraphics[width=\linewidth]{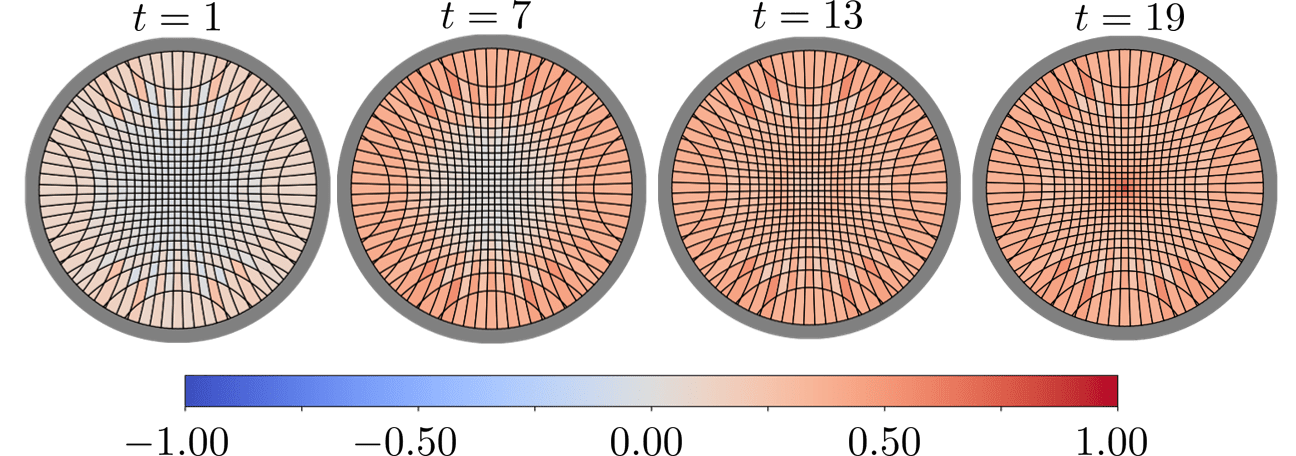}  
  \caption{Component $h_{11}$.}
  \label{fig:sub-therm-h11}
  \vspace{0.5cm}
\end{subfigure}
\begin{subfigure}{0.7\textwidth}
  \includegraphics[width=\linewidth]{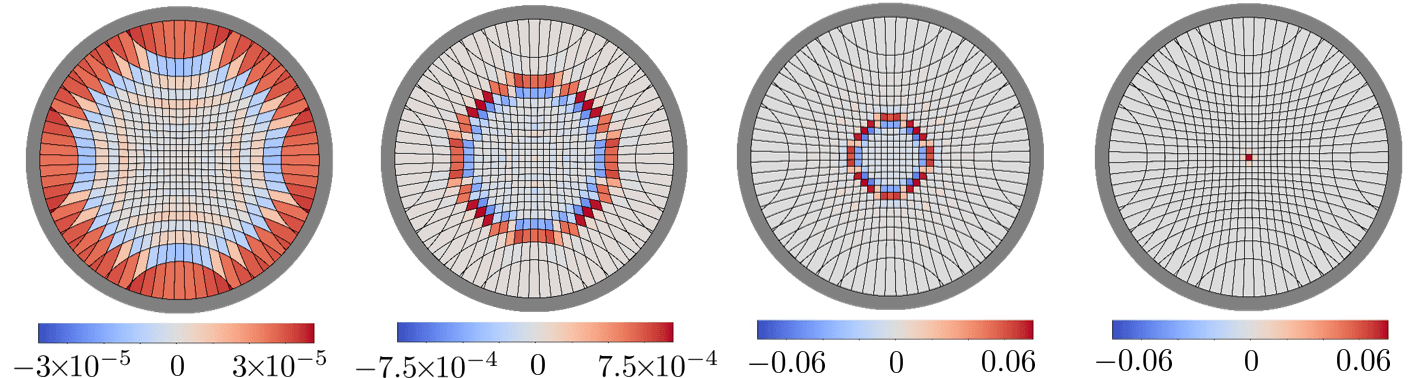}  
  \caption{Linearized scalar curvature perturbation. }
  \label{fig:sub-therm-ricci}
\end{subfigure}
\caption{Metric perturbation from volume law entropy growth~\eqref{eq:vollaw}. Plots are ordered from left to right as time increases. We only give $h_{11}$ for the sake of clarity, because $h_{12}\approx 0$ and $h_{22}\approx h_{11}$. The boundary relative error is $\cE_{\mathrm{bdy}}\approx 0.03$.}
\label{fig:thermalization}
\end{figure}

As the system thermalizes, larger subregions have a volume law entropy. The wavefront of the entanglement spread is reflected in the bulk as a spherically symmetric perturbation moving from the boundary to the center. Assuming Einstein gravity, which holds for holographic CFTs, the curvature perturbation $\delta R$ also reflects the bulk matter distribution through the linearized Hamiltonian constraint for each instance of time \cite{Czech:2016tqr}. Therefore, this thermalization process is consistent with the collapse of a spherical shell of matter~\cite{Liu:2013qca}. 

\subsubsection{Mixture of Thermal States}
There are also instances where we do not expect a well-defined geometry to exist. For example, when the state is taken from a theory where the bulk is strongly coupled and/or the state is a macroscopic superposition of certain classical geometries, quantum gravitational effects can dominate, leading to a breakdown of the classical geometric description.

In this section, as an example of a potentially non-geometric holographic state, we look at mixtures of thermal AdS geometries at various temperatures. We will consider states of the form
\begin{equation}
\rho = p\rho(T_1) + (1-p) \rho(T_2),
\end{equation}
where $\rho(T_i)$ are thermal states of a CFT with distinct temperatures $T_1$ and $T_2$. From~\cite{Almheiri:2016blp}, the von Neumann entropy of the mixture is estimated as
\begin{equation}
    S(\rho_A) = S(\rho_{1,A})+S(\rho_{2,A})+H(p), 
\end{equation}
where we write $\rho_{i,A}$ to denote the reduced state of $\rho(T_i)$ on a boundary region $A$, and where
\begin{equation}
    H(p) = -p\log p -(1-p)\log(1-p)
\end{equation}
is a Shannon-like term that corresponds to the entropy of mixing~\cite{Almheiri:2014lwa}. Physically, such a state can be created by superposing two thermofield double states at different temperatures, and then tracing out one side of the wormhole. For CFTs with strong bulk gravity where $G_N \approx 1$, we find the superposition incurs a large error, indicating non-geometric configurations in the bulk. However, for a CFT with a weakly coupled dual where $G_N\ll 1$, the geometry smoothly interpolates between the two temperatures at the linearized level, consistent with our expectation that the entropy operator is proportional to the area operator at leading order in $N$. In the language of Radon transform, this is caused by the entropy of mixing $H(p)$ being a non-geometrical contribution.

\begin{figure}[H]
    \centering
    \includegraphics[width=0.7\textwidth]{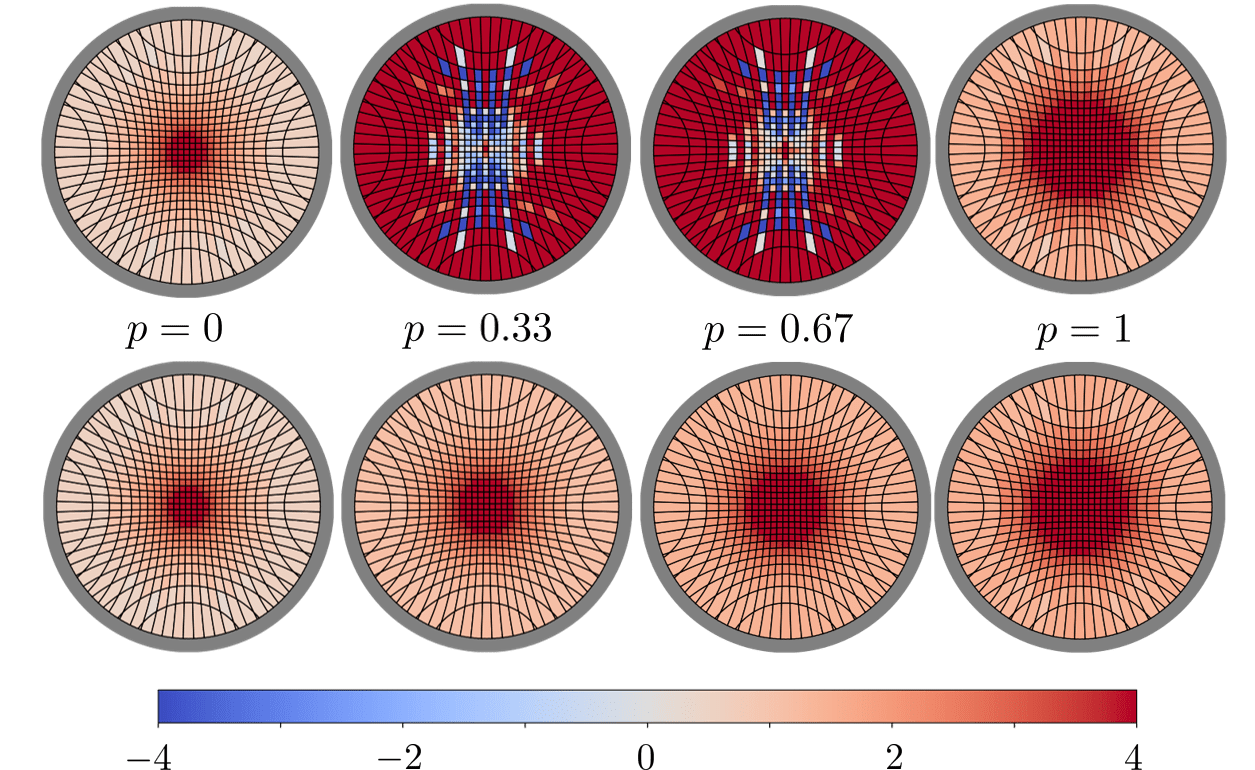}
    \caption{Component $h_{11}$ of a mixture of two thermal geometries at two distinct temperatures. Top diagram for $G_N \approx 1$ and bottom for $G_N\ll 1$, both with $L_{\rm AdS} = 1$. From left to right, we increase the mixing ratio $p$ and the geometry transitions smoothly from one at lower temperature to the other at higher temperature when the gravitational coupling is weak. The reconstruction is dominated by artifacts and have large error if the coupling is strong, leading to a contribution $G_N H(p)$ in the leading order.}
    \label{fig:thermalsuperpose}
\end{figure}

\subsection{1D Free Fermion}

Given that the reconstruction procedure relies only on the von Neumann entropy of a state, we may also apply it to various quantum many-body systems where we do not necessarily expect the conformal field theory to be dual to a semi-classical bulk theory with weakly coupled gravity. The lack of a semi-classical geometric dual will generally be reflected in the presence of a large boundary reconstruction error. One example where the separation between geometric and non-geometric states is particularly manifest is the contrast between the reconstruction of holographically motivated data and that of a $1$d free fermion.

Generic low-energy states for free fermion systems are generally believed to have a highly quantum bulk because the conformal field theory is non-interacting, and has a small central charge. Heuristically, the strong-weak nature of the holographic duality suggests that such systems should result in a strongly coupled bulk geometry, if such geometries even exist. In this section, we give a few examples indicating that generic excited states in the free fermion system are indeed non-geometric. However, we shall also see that certain large-scale geometric properties may nevertheless be present, even if the overall state is ostensibly non-geometric.

For our reconstruction, we consider a $100$-site massless free fermion system with periodic boundary conditions, and with Hamiltonian
\begin{equation}
    \hat{H}_0 = -\sum_i \hat{a}^{\dagger}_i \hat{a}_{i+1}+ \mathrm{h.c.},
    \label{eqn:ffham}
\end{equation}
where, as usual, the creation and annihilation operators satisfy the canonical anti-commutation relations 
\begin{align}
\{\hat{a}_i, \hat{a}_j^{\dagger}\}=\delta_{ij},\quad \{\hat{a}_i, \hat{a}_j\} = 0,\quad \{\hat{a}_i^{\dagger}, \hat{a}_j^{\dagger}\}=0.
\end{align}
We shall also consider various quench dynamics and deformations of the free fermion system.

The free fermion system is described by a $(1+1)$-d conformal field theory in the thermodynamic limit, with central charge $c=1/2$ \cite{DiFrancesco:639405}. Here we study directly the regulated lattice model. Any such model of non-interacting fermionic modes can be studied efficiently numerically with a cost scaling polynomially with the number of fermion modes $N$. In particular, we can compute all of the single interval entanglement entropies via~\cite{PeschelEisler2009}. The key point is that the state of the system is Gaussian. In particular, given a subset $A$ of the fermions, the reduced density matrix is $\rho_A \propto e^{- K_A}$ with $K_A$ quadratic in $\hat{a}$ and $\hat{a}^\dagger$,
\begin{equation}
    K_A = \hat{a}^\dagger k_A \hat{a}.
\end{equation}
Here $k_A$ is an $|A| \times |A|$ Hermitian matrix that determines all correlation functions of fermions in $A$. It is related to the 2-point function via
\begin{equation}
    G^A_{ij} = \langle \hat{a}_i^\dagger \hat{a}_j \rangle \big|_{i,j \in A} = \left(\frac{1}{e^{k^{\mathbb{T}}_A}+1} \right)_{ij}.
\end{equation}
The entropy of $A$ is then determined by the eigenvalues of $k_A$, but there is also a direct formula in $G^A$:
\begin{equation}
    S(A) = - \text{tr}\left[ G^A \ln G^A  + (1-G^A) \ln (1-G^A) \right].
\end{equation}

To fix our background geometry, we numerically compute the ground state entanglement of the critical Hamiltonian $\hat{H}_0$. To consider non-vacuum emergent geometries, we consider states $\ket{\psi} \ne \ket{0}$, which are not necessarily energy eigenstates of $\hat{H}_0$. For instance, these excited states can be generated by first deforming the Hamiltonian away from criticality through a mass deformation
\begin{align}
\hat{H}_0 \mapsto \hat{H}_0 +m\hat{H}_1,
\end{align}
and then finding the ground state $|\psi_m\rangle$ of the perturbed Hamiltonian. For small $m$, the new ground state $\ket{\psi_m}$ is generically an excited state with relatively low energy in the unperturbed system.

We will study reconstructions in both the static and dynamical cases. For the static case, we reconstruct the emergent geometry of the new ground states $\ket{\psi_m}$ by finding the best-fit geometry using the discrete Radon transform. We do this for a number of distinct values of $m$. For the dynamical case, we consider the quench dynamics corresponding to the deformation $\hat{H}_1$. We start with some fixed deformed ground state $\ket{\psi_m}$ and then time evolve the state with the free Hamiltonian $\hat{H}_0$ and reconstruct the geometry that corresponds to each time step: $\ket{\psi_m(t)} = \exp(-it\hat{H}_0)\ket{\psi_m}$.

\subsubsection{Local Deformations}
In this section and the next, we will consider deformations of the form
\begin{equation}
    \hat{H}_1 = \sum_{i\in S}\hat{a}^{\dagger}_i\hat{a}_{i+1}+\mathrm{h.c.},\label{eq:mass_deformation}
\end{equation}
where $m$ is a small positive parameter, and where $S$ is a set of sites where we introduce such deformations. The perturbation serves to deform the ground state wavefunction around the sites near $S$. We label the sites counter-clockwise from $0$ to $99$, starting with site $0$ aligned with the positive $x$-axis (i.e., site $i$ sits on the unit circle at angle $\theta_i = 2\pi i/100$). We will consider local perturbations where $S$ consists of one or two sites in this section, and more global perturbations in the next section.

Figure~\ref{fig:singlesite_nosmooth} shows the best-fit geometry arising from the ground state $\ket{\psi_m}$ corresponding to a local mass deformation at a single site $S=\{0\}$, located along the positive $x$-axis. We can see that, in contrast to the holographic reconstructions in section~\ref{subsec:holorecon}, the geometries shown in Figure~\ref{fig:singlesite_nosmooth} are heavily dominated by noise and local artifacts. The corresponding relative error of reconstruction is also larger by more than an order of magnitude as compared to the holographic cases (see Figure~\ref{fig:histogram_err} in section~\ref{sec:geodetect} for a more comprehensive comparison of relative errors). The large error of reconstruction can be seen as a hallmark of the fact that the underlying state is non-geometric.

\begin{figure}[H]
\centering
  \includegraphics[width=0.7\linewidth]{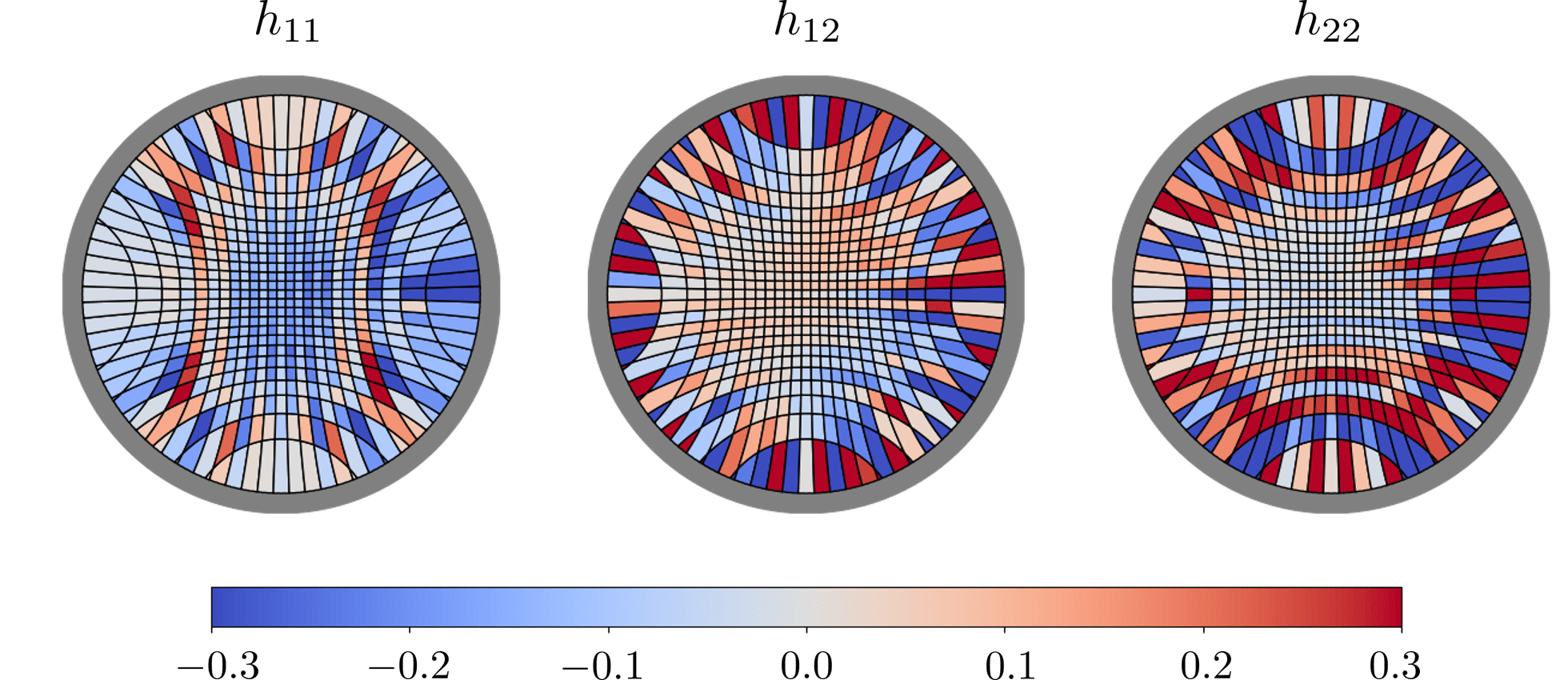}  
  \caption{Components of the reconstructed metric tensor perturbation corresponding to the ground state of a mass deformed Hamiltonian, with the deformation located at a single site located along the positive $x$-axis. The boundary relative error is $\cE_{\mathrm{bdy}} \approx 0.62$.}
  \label{fig:singlesite_nosmooth}
\end{figure}
Note that the lack of geometric features in the reconstruction does not necessarily indicate that the reconstruction procedure is flawed. In fact, non-geometric features are expected because not all boundary data should produce geometric reconstructions. Our reconstruction procedure has equipped us with information of telling apart when that will happen. We will discuss this further in section~\ref{sec:geodetect}. 

Nevertheless, there are cases where some large-scale geometric features can be extracted from the plot. This is especially clear in the dynamical case, where we consider the quench dynamics obtained by evolving the deformed ground state $\ket{\psi}$ using the free fermion Hamiltonian $\hat{H}_0$. The reconstruction, performed timeslice by timeslice, is shown in Figure~\ref{fig:singlesite_nosmooth_quench} (only the dominant component $h_{11}$ is shown). While it may not be entirely clear, Figure~\ref{fig:singlesite_nosmooth_quench} reveals a large scale shockwave that originates from the deformation site and then travels across the bulk before being reflected at the left boundary.

\begin{figure}[H]
\centering
  \includegraphics[width=0.7\linewidth]{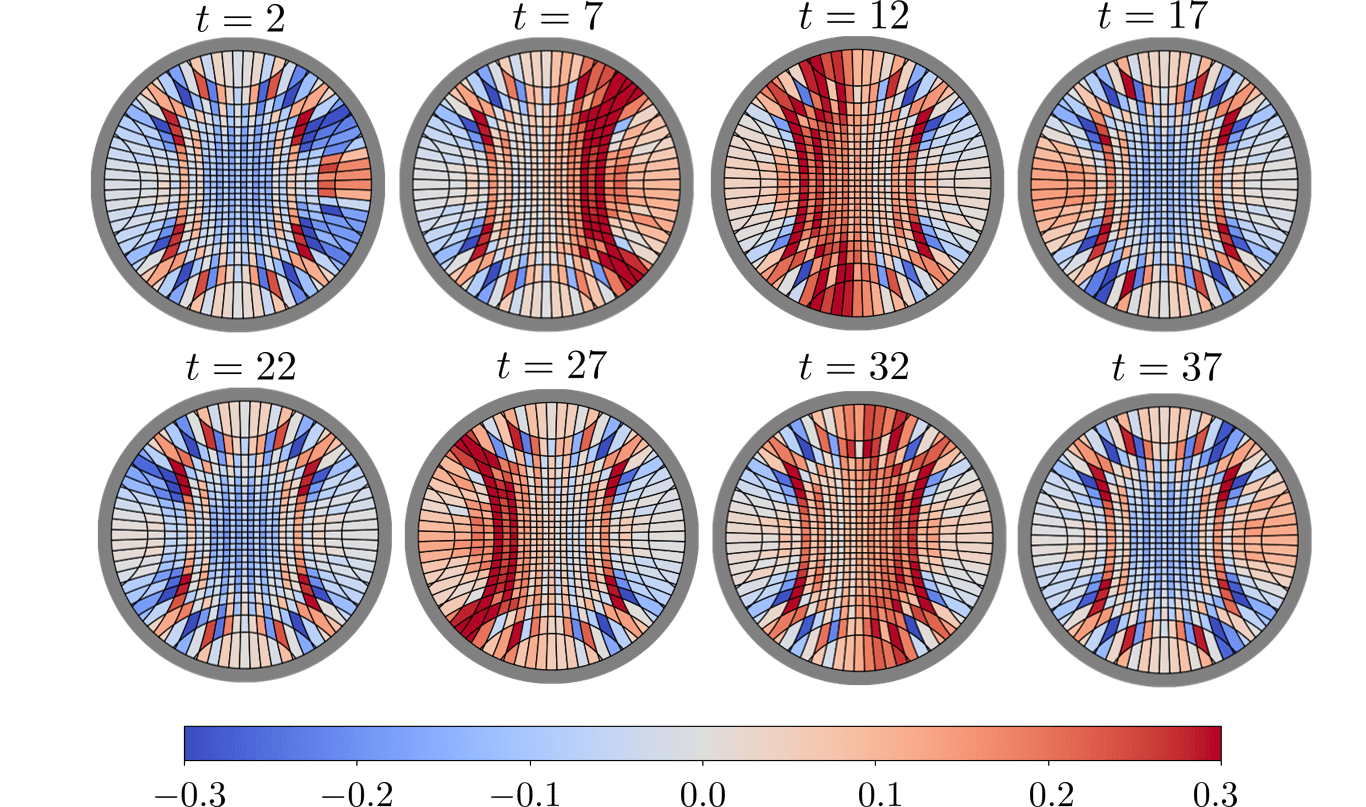}  
  \caption{Quenched time dynamics of the reconstructed geometry corresponding to a single mass deformation along the $x$-axis (see Figure~\ref{fig:singlesite_nosmooth} for the $t=0$ geometries). Only the dominant component $h_{11}$ is shown here. The average reconstruction error over time is $\cE_{\mathrm{bdy}}\approx 0.52$.}
  \label{fig:singlesite_nosmooth_quench}
\end{figure}

The large-scale geometry is heavily obscured by non-geometric noise in Figure~\ref{fig:singlesite_nosmooth_quench}. 
To better extract the large-scale features that we consider to be relevant, we must filter out the small-scale artifacts what we consider to be ``noise''. 

To more clearly extract the underlying large-scale features of what appears to be waves generated by the perturbations, we perform some pre-processing of the boundary data to smooth out the small-scale noise. This is done by averaging the entanglement entropies of two neighbouring intervals of the same size. More precisely, suppose $S(i,j)$ is the von Neumann entropy on the interval from site $i$ to site $j$. Then the smoothing procedure we applying is similar to a filter by convolving with a 2-site window function such that 
\begin{equation}
    S(i,j) \rightarrow \frac{S(i,j)+S(i+1,j+1)}{2}.
\end{equation}

This produces a less noisy reconstruction and significantly reduces the amount of non-geometric contributions to boundary relative errors. Intuitively, the smoothing procedure acts as a low-pass filter in the space of boundary intervals. This serves to remove the short distance artifacts normally associated with non-geometricality. While we used a specific filter in this example, one may apply more general filter constructions for other purposes.\footnote{It is reasonable to suspect that the two site averaging of entanglement entropy works well as a filter here because UV details in the free fermion Hamiltonian (\ref{eqn:ffham}) with Fermi momentum $\pm \pi/2$ contributed to non-geometric noise. We might therefore expect that we can remove such non-geometric contributions through coarse-graining by grouping together adjacent sites in pairs. While this grouping does indeed remove some non-geometric noise, it does not remove it entirely: the reconstruction error after grouping is $0.1\lesssim \mathcal{E}\lesssim 0.3$, which is in between what is shown in Figure~\ref{fig:singlesite_nosmooth_quench} and Figure~\ref{fig:singlesite_smooth_quench}. Given that the smoothed data reconstruction error is still significantly larger than that of completely geometric data, it is not clear if the non-geometric contribution completely reside in the UV. }

We wish to emphasize that this filtering procedure is not a part of our reconstruction process. It is merely here to assist us in gaining some intuition regarding the qualitative physical behaviour of this type of dynamical process. One could imagine that the filtering reveals what the bulk physical process might have looked like, had the state actually been geometric.

The corresponding smoothed geometries are shown in Figure~\ref{fig:singlesite_smooth_quench}. It can be seen that smoothing significantly reduces local noise, and clearly reveals large-scale time dynamics associated with the quench that looks like a (shock)wave.

\begin{figure}[H]
\centering
  \includegraphics[width=0.7\linewidth]{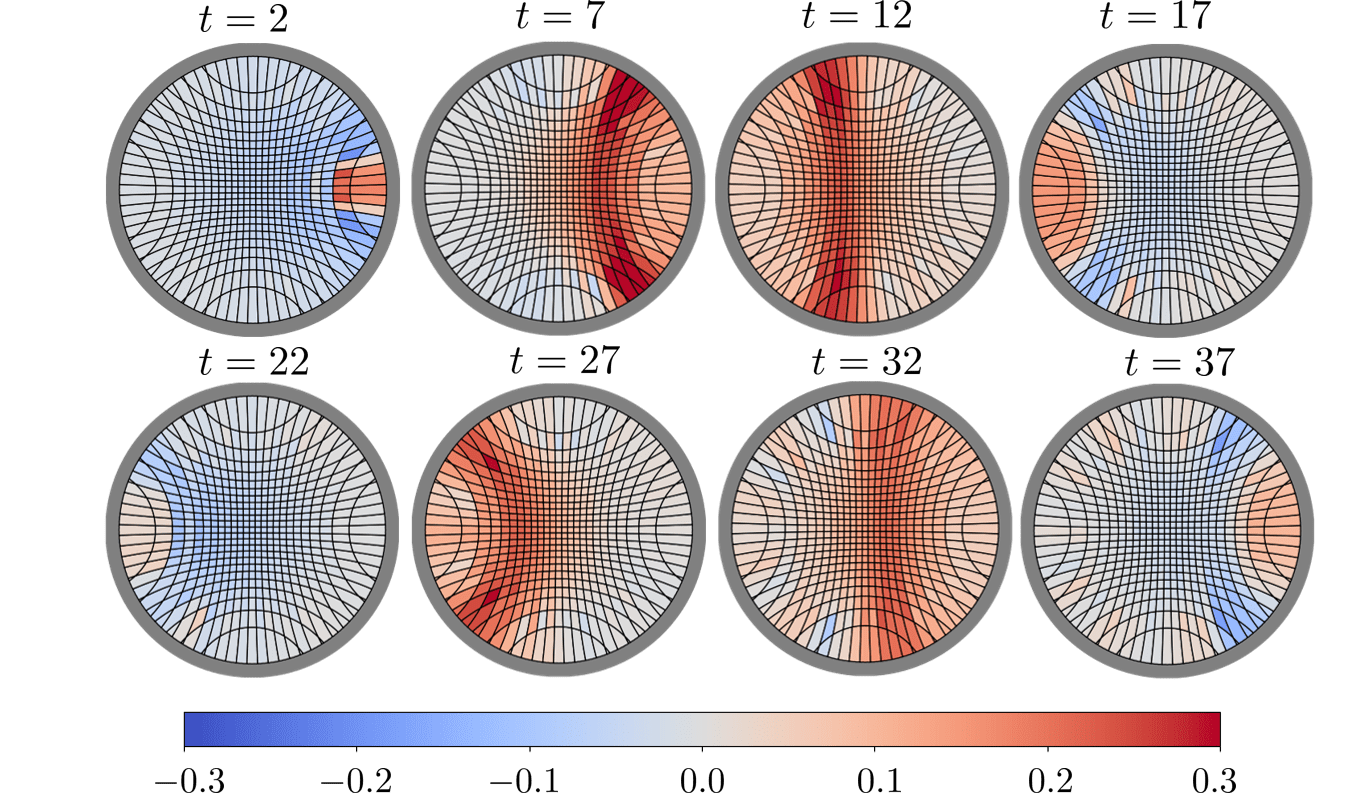}  
  \caption{Quenched time dynamics of the reconstructed geometry corresponding to a single mass deformation along the $x$-axis. The boundary data is pre-processed by smoothing to remove small-scale details (see Figure~\ref{fig:singlesite_nosmooth_quench} for the corresponding unsmoothed reconstruction). Only the dominant component $h_{11}$ is shown here. The relative error is $0.03 \lesssim\cE_{\mathrm{bdy}} \lesssim 0.1$.}
  \label{fig:singlesite_smooth_quench}
\end{figure}

As another example, the quench dynamics of a state with two distinct deformations at sites $S=\{0,30\}$ is shown in Figure~\ref{fig:twosite_unsmooth_quench} (unsmoothed) and Figure~\ref{fig:twosite_smooth_quench} (smoothed).

\begin{figure}[H]
\centering
  \includegraphics[width=0.7\linewidth]{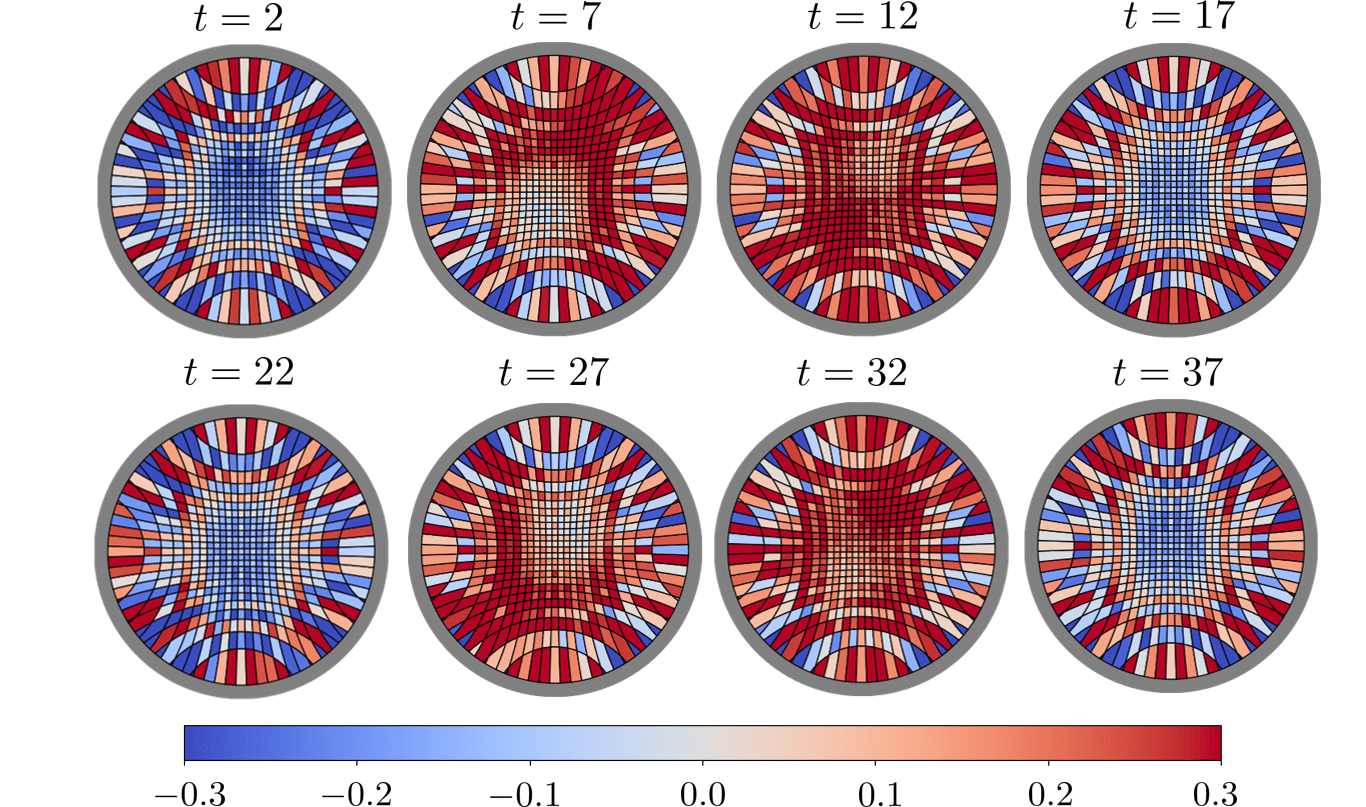}  
  \caption{Quenched time dynamics of the reconstructed geometry corresponding to a mass deformation at two distinct sites ($S = \{0,30\}$). The trace of $h_{ij}$ is shown here. The boundary relative error is $\cE_{\mathrm{bdy}} \approx 0.41$.}
  \label{fig:twosite_unsmooth_quench}
\end{figure}

\begin{figure}[H]
\centering
  \includegraphics[width=0.7\linewidth]{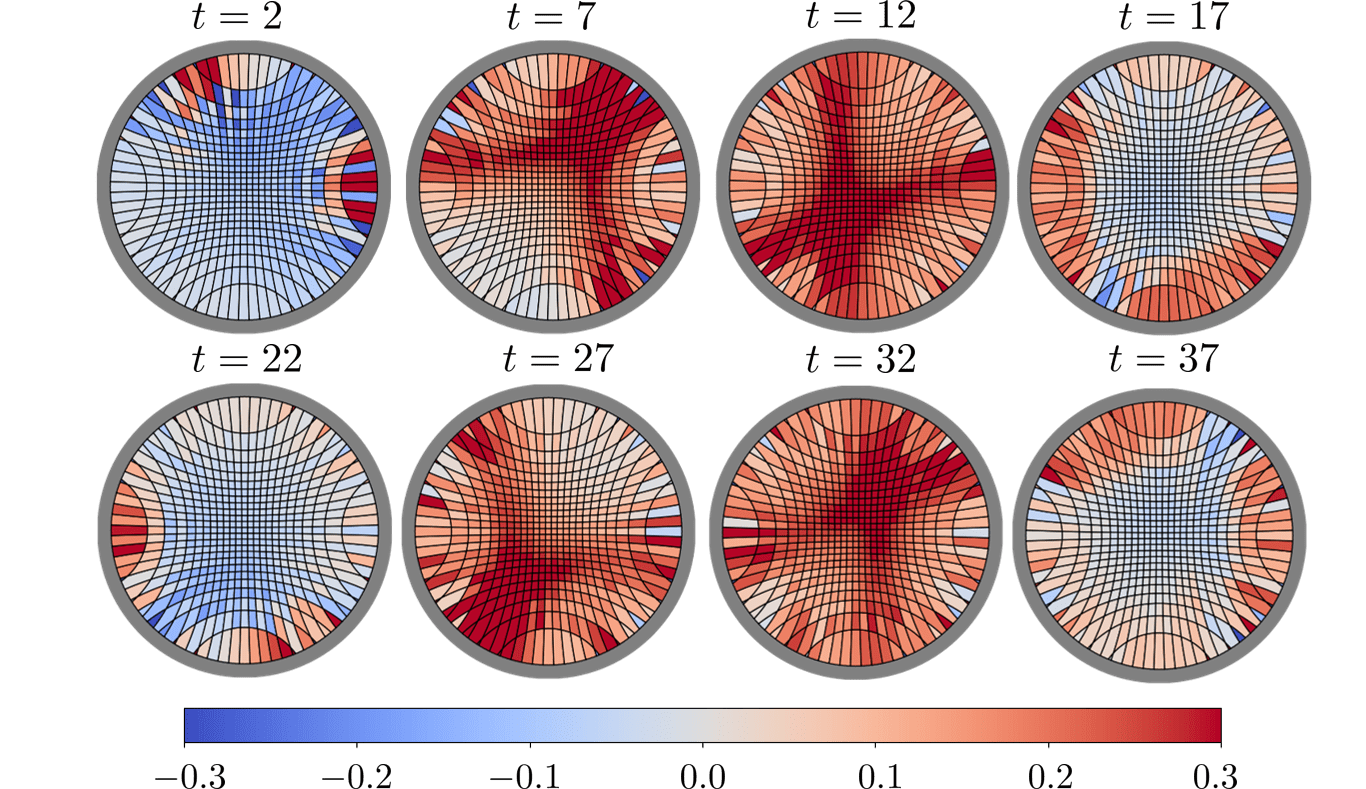}  
  \caption{Quenched time dynamics of the reconstructed geometry corresponding to a mass deformation at two distinct sites ($S = \{0,30\}$). The boundary data is pre-processed by smoothing to remove small-scale details. The trace of $h_{ij}$ is shown here. The boundary relative error is $\cE_{\mathrm{bdy}} \approx 0.05$.}
  \label{fig:twosite_smooth_quench}
\end{figure}

In all cases, we see that the free fermion Hamiltonian will generate seemingly non-interacting waves that traverse through the hypothetical bulk spacetime. Since the underlying time dynamics is integrable, the same entanglement feature recurs after the waves traverse the entire system; we show one such iteration in our figures.

\subsubsection{Global Deformations}
Similar to the previous analysis for local deformations, we can also consider global deformations in the same vein. In this section, we will consider Hamiltonian perturbations of the same form as~\eqref{eq:mass_deformation}, but now with deformations located at every other site, i.e., $S=\{0,2,4,\cdots, 98\}$.

In Figure~\ref{fig:global_massdef}, we show the geometries reconstructed from the ground state $\ket{\psi_m}$ of a Hamiltonian with the aforementioned global deformations, plotted across a range of $m$ values. Again, we find the overall geometry to be highly dominated by noise, with a large relative error of reconstruction indicating that the underlying state is non-geometric. The relative error of reconstruction generally becomes worse with increasing values of $m$ (see Figure~\ref{fig:mass_def_err}). Qualitatively, the reconstructions for the different values of $m$ appear similar, the main difference being the overall magnitude of the metric perturbation.

\begin{figure}[H]
\centering
  \includegraphics[width=0.7\linewidth]{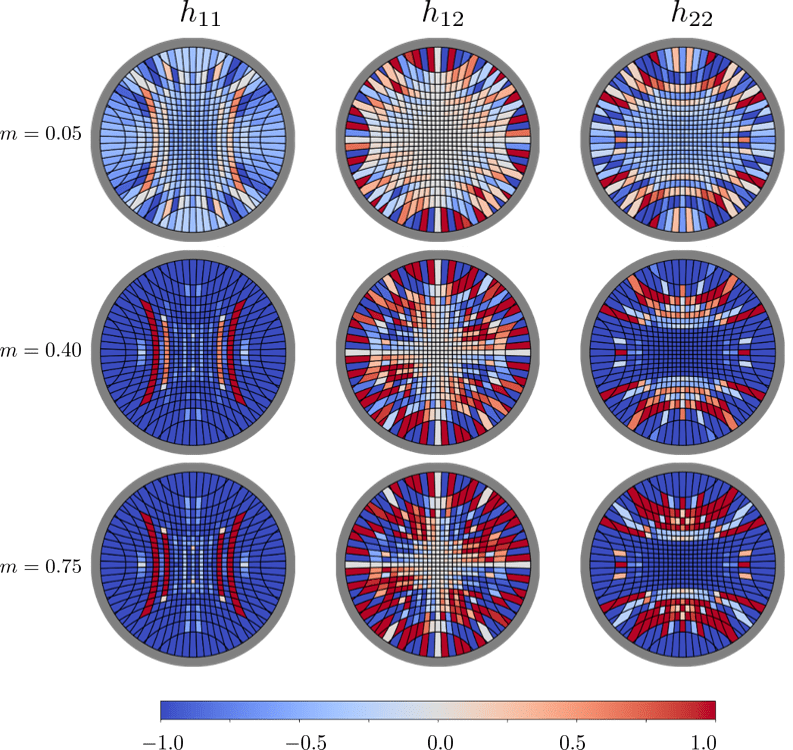}  
  \caption{Components of the reconstructed metric tensor perturbation corresponding to the ground state of a globally mass deformed Hamiltonian (with $m$ values as shown), with deformations located at every other site. The relative errors are approximately $\cE_{\mathrm{bdy}} \approx 0.25$ for this set of plots (see Figure~\ref{fig:mass_def_err}).}
  \label{fig:global_massdef}
\end{figure}

Looking at the corresponding quench dynamics (see Figure~\ref{fig:global_massdef_quench_unsmooth} for the unsmoothed reconstruction and Figure~\ref{fig:global_massdef_quench_smooth} for the smoothed version), we see that the large-scale geometry involves a spread of entanglement that is qualitatively similar to the configuration obtained in the thermalization scenario considered in section~\ref{subsec:holorecon}. However, in this case the bulk ``matter'' is non-interacting. The integrability of free fermion Hamiltonian ensures that the falling shockwave returns to the boundary after some finite time instead of collapsing into a steady configuration.

\begin{figure}[H]
\centering
  \includegraphics[width=0.7\linewidth]{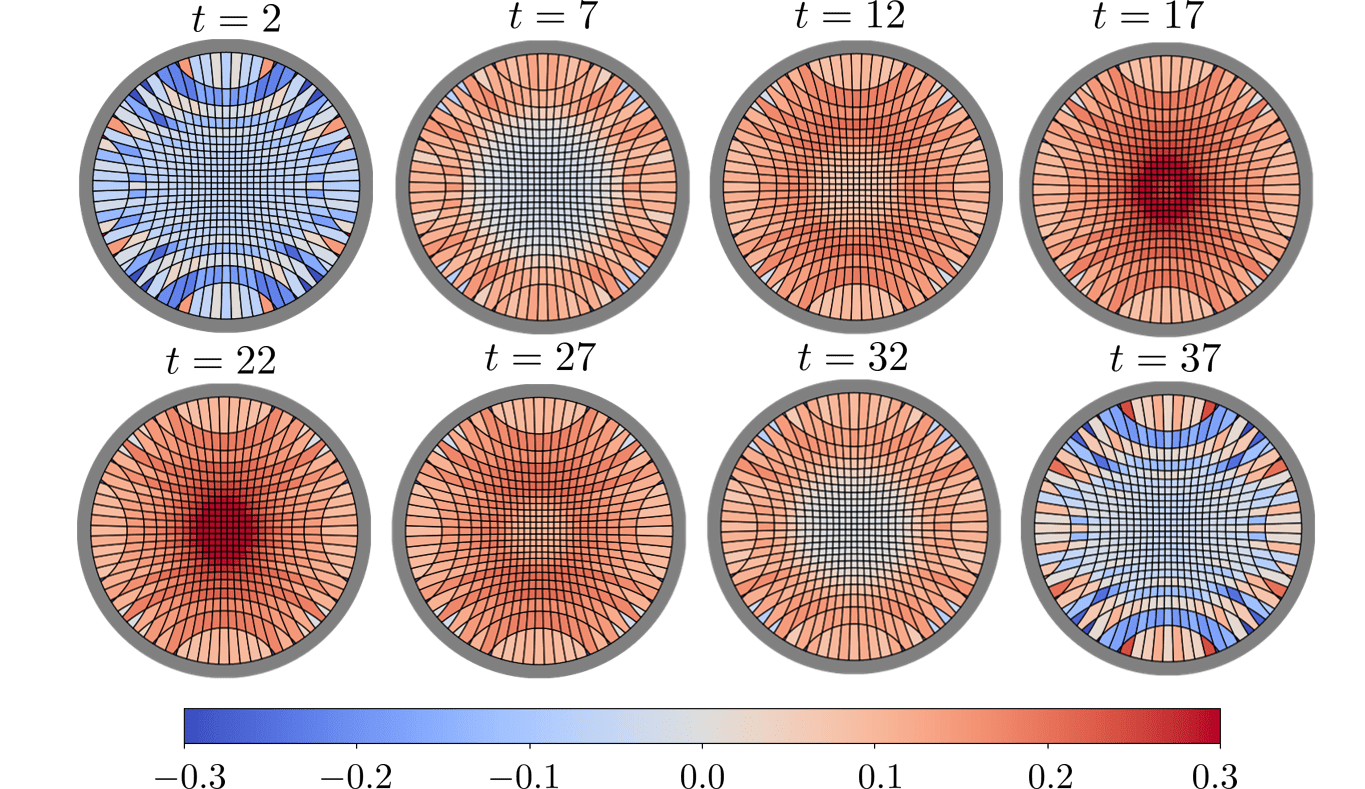}  
  \caption{Quenched time dynamics of the reconstructed geometry corresponding to a global mass deformation (see Figure~\ref{fig:global_massdef} for the $t=0$ geometries). The component $h_{22}$ is shown here. The boundary relative error is maximal $\cE_{\mathrm{bdy}}\approx 0.22$ near $t=0$ and is minimal $\cE_{\mathrm{bdy}} \approx 0.006$ around $t=20$.}
  \label{fig:global_massdef_quench_unsmooth}
\end{figure}

\begin{figure}[H]
\centering
  \includegraphics[width=0.7\linewidth]{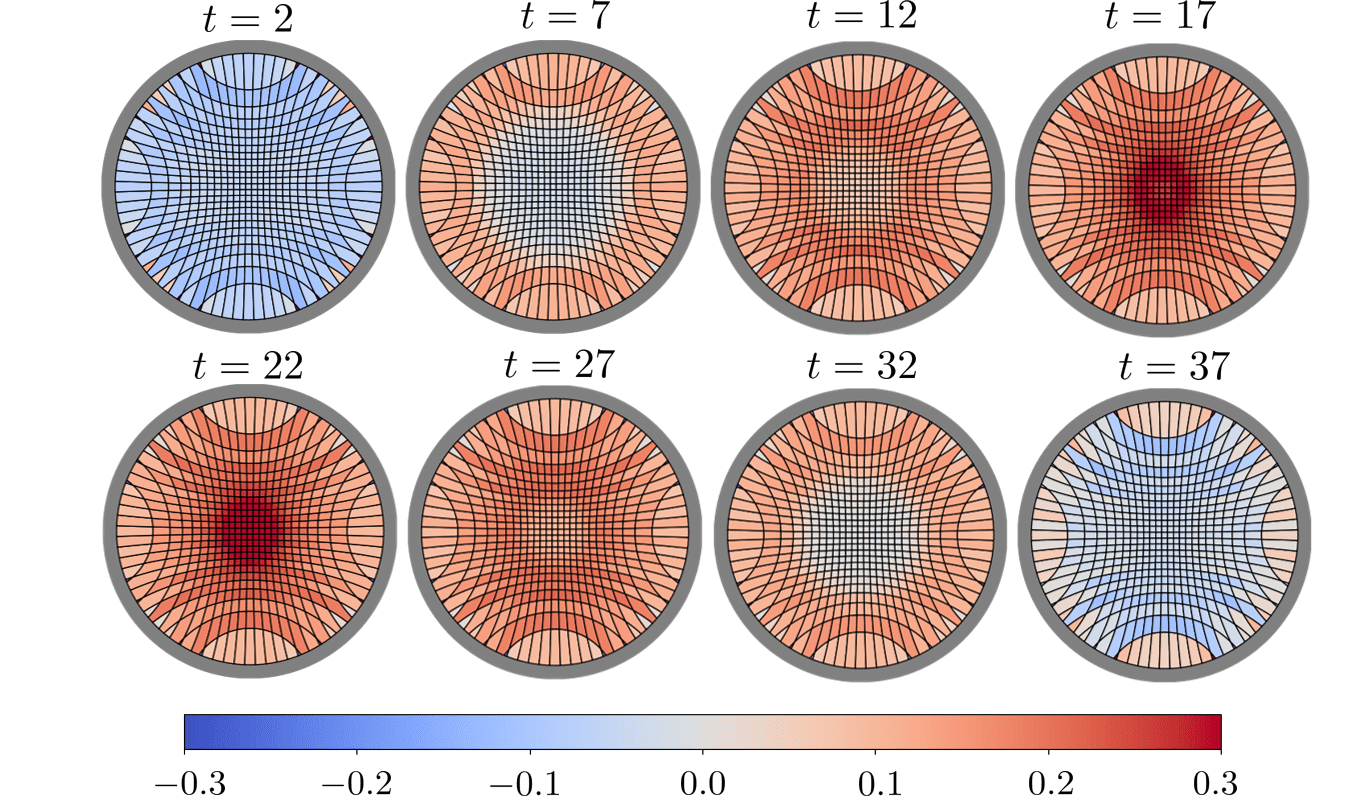}  
  \caption{Quenched time dynamics of the reconstructed geometry corresponding to a global mass deformation. The boundary data is pre-processed by smoothing to remove small-scale details. The dominant component $h_{22}$ is shown here. The boundary relative error is largely similar to the unsmoothed version except at the initial/final times, where it is significantly reduced by the smoothing.}
  \label{fig:global_massdef_quench_smooth}
\end{figure}

\subsection{Random Disorder}

Finally, we can take a look at the geometries arising from a random perturbation of the free fermion system. Disorder is introduced by adding a random external field at each site of the spin chain
\begin{equation}
    \hat{H}_{\rm disorder} = \hat{H}_0+\sum_i w_i\hat{a}^{\dagger}_i\hat{a}_i,
\end{equation}
where each parameter $w_i$ is a random parameter chosen i.i.d. from a uniform random distribution over the interval $[-0.1, 0.1]$.

In Figure~\ref{fig:random_disorder}, we show a generic sampling of ground states obtained from such random disorder. As would be naively expected, the resulting ground states have generically large relative errors, with no discernible large scale features (quench dynamics also reveal no discernible large-scale patterns) or symmetries.

\begin{figure}[H]
\centering
\includegraphics[width=0.7\linewidth]{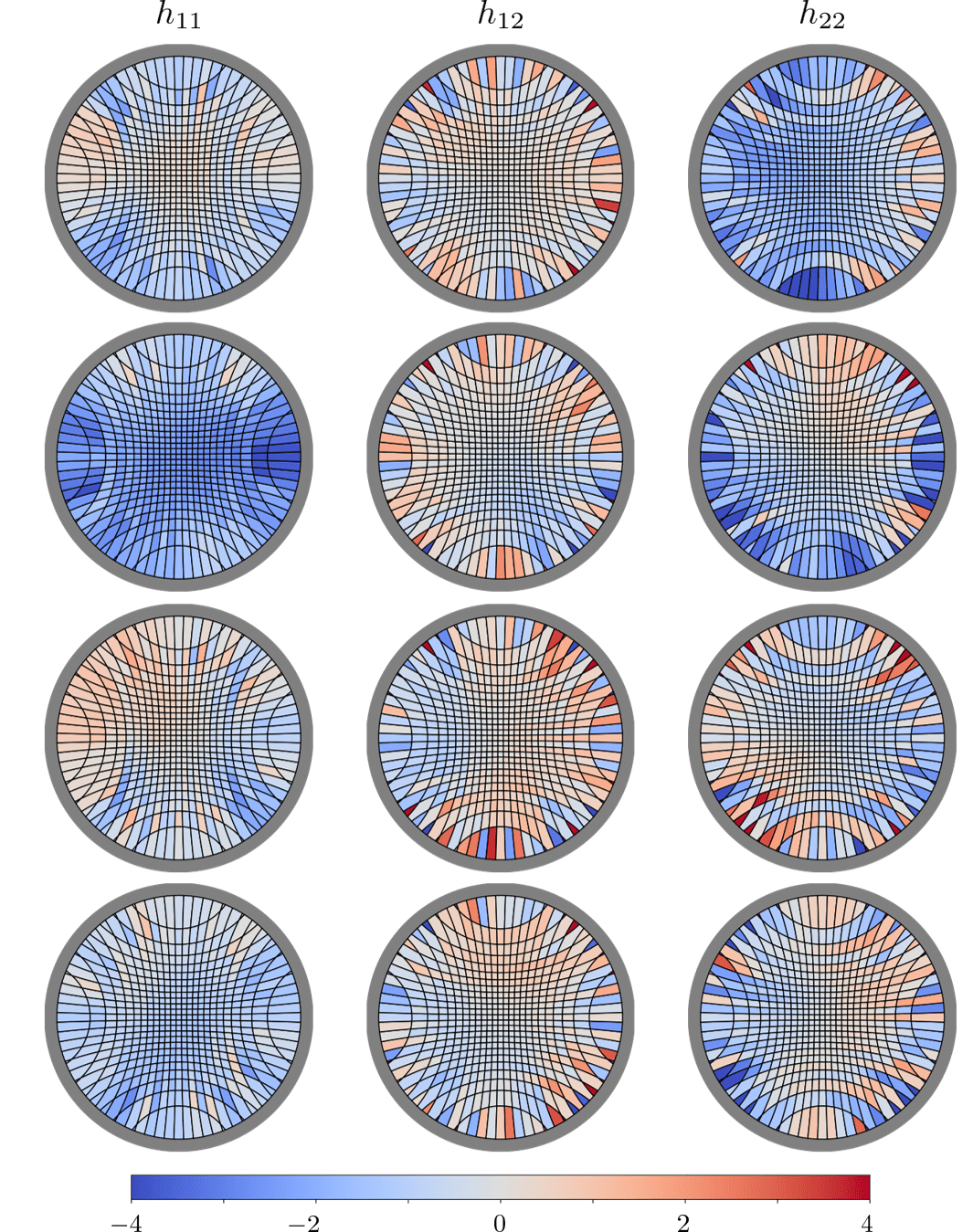}  
  \caption{A bulk best-fit reconstruction of a state generated with random disorder. Each row is a reconstruction of a particular instance with random disorder. The boundary relative errors vary across a wide range of values, but typically $0.05\lesssim \cE_{\mathrm{bdy}} \lesssim 1$ (see Figure~\ref{fig:histogram_err}).}
  \label{fig:random_disorder}
\end{figure}

\section{Geometry Detection}
\label{sec:geodetect}
In this section, we summarize the results of the previous reconstructions and comment on their similarities and differences. We also emphasize that a small relative error of reconstruction is indicative of a classical bulk geometry in which boundary entropies are computed via the RT formula, whereas a large error may indicate some combination of (1) no classical geometry, (2) no classical geometry near the AdS background, or (3) a classical geometry but where entropies have non-trivial contributions from sources other than the area of RT surfaces, e.g. higher order corrections. For the purpose of this discussion, we refer to all three of these negative cases as non-geometric, but it would clearly be desirable to distinguish them further in future work.

From the results of the mass deformation, we confirm that generic low-energy states of a free fermion system do not appear to have a good geometric reconstruction. Comparing the ground states of the mass deformed $1$d free fermion with that of the thermal AdS state, we see that the relative errors for the $1$d free fermion are significantly larger than those of the holographic data (see Figures~\ref{fig:histogram_err},\ref{fig:mass_def_err},\ref{fig:quench_err_time}). Smoothing of the boundary data reduces the relative error of reconstruction by a significant amount (as would be expected due to the reduction of small-scale artifacts). However, even with smoothing, the level of error is clearly distinguishable from the reconstructions of known geometric states such as thermal AdS.

In particular, ground states of the locally deformed $1$d free fermion Hamiltonian have a relative error that is of order $\sim 1$. Smoothing of the boundary data brings this down to $\sim 10^{-1}$, which is still an order of magnitude larger in comparison to the thermal AdS reconstructions, which have a relative error on the order of $\sim 10^{-2}$. The unsmoothed relative error for the global deformation of the $1$d free fermion hovers around $\sim 10^{-1}$. Smoothing brings this down significantly to $\sim 10^{-2}$. The most likely explanation for the pronounced effect of smoothing here is due to the global symmetry present in the globally deformed state. We also see that the relative errors for reconstructions corresponding to random disorder in the $1$d free fermion tends to interpolate between the results for local and global deformations across different random instances, with the best case relative errors being $\lesssim 10^{-1}$ and the worst case being $\sim 1$. A summary of these results is plotted in Figure~\ref{fig:histogram_err}, which shows a histogram of the number of instances for each type of reconstruction, as we vary some of the relevant parameters (i.e., temperature, mass, etc). In Figure~\ref{fig:mass_def_err}, we also plot the relative error of reconstruction as a function of the mass deformation parameter $m$. As would be expected, we see that a larger mass deformation, and hence a larger deviation from criticality, contributes positively to the relative error of reconstruction regardless of smoothing. 

\begin{figure}[H]
\begin{subfigure}[t]{0.48\textwidth}
  \centering
  \includegraphics[width=\linewidth]{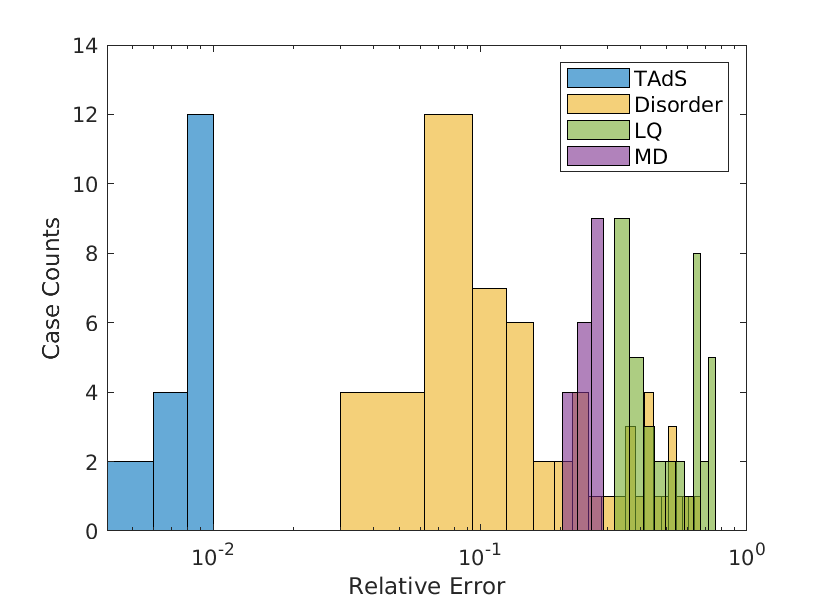}
  \caption{A histogram of relative errors from different boundary data: Thermal AdS (TAdS), global mass deformation (MD), local quench without smoothing (LQ), and system with random disorder (Disorder). Boundary relative error $\cE_{\mathrm{bdy}}$ plotted on the $x$-axis. }
  \label{fig:histogram_err}
\end{subfigure}
\hfill
\begin{subfigure}[t]{0.48\textwidth}
  \centering
  \includegraphics[width=\linewidth]{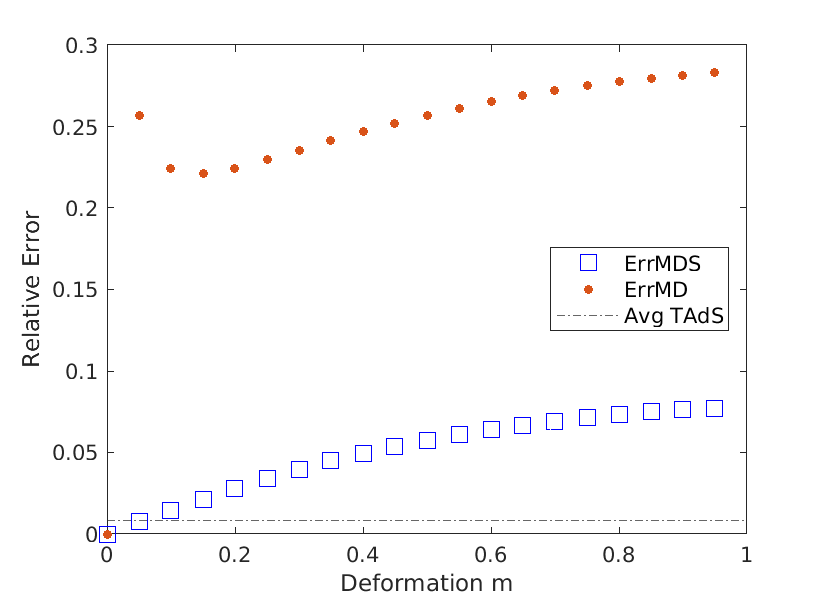}  
  \caption{Global mass deformation errors, as a function of the deformation parameter $m$. Red dots denote the boundary relative error of the original fermion data. Blue squares are relative errors after smoothing.}
  \label{fig:mass_def_err}
\end{subfigure}
\caption{}
\end{figure}

A summary of the time dynamics of the relative errors is shown in Figure~\ref{fig:quench_err_time}. For the dynamical reconstructions, we see that the relative error generally varies with time.  This is most noticeable with the global quench dynamics, where the relative error tends to decrease as entanglement spreads. Although one starts with a mass deformed state with large relative error at $t=0$, subsequent evolution can effectively wash out non-geometric features the system begins to thermalize. The state at $t\approx 20$ captures basic entanglement properties of a thermal state, thus the reconstruction is thermal AdS-like with small relative errors. Nevertheless, the state cannot actually thermalize due to integrability; the system recurs at later times and the relative error rebounds. 

At this point, it is not completely clear why the global quench states have such small relative error at intermediate times. A likely guess that the enhanced symmetries for the globally deformed states contribute to a more geometrical reconstruction, with the effect being especially pronounced at intermediate times when the system is maximally thermal, before it is kicked back by integrability. Such behaviour is not seen in the local quench disorder, where the magnitude of the relative error is larger, and remains  stable across the time evolution. However, symmetry is likely not the full answer either, since otherwise we would expect the disordered models to generically have larger relative error compared to the global mass deformed state, whereas we observer the disordered model to be peaked at smaller errors. While the full characterization for which states appear more geometric than others is currently unknown, it seems likely that there are many contributing factors to a successful reconstruction, among them properties like symmetry and thermality.

\begin{figure}[H]
  \centering
  \includegraphics[width=0.6\linewidth]{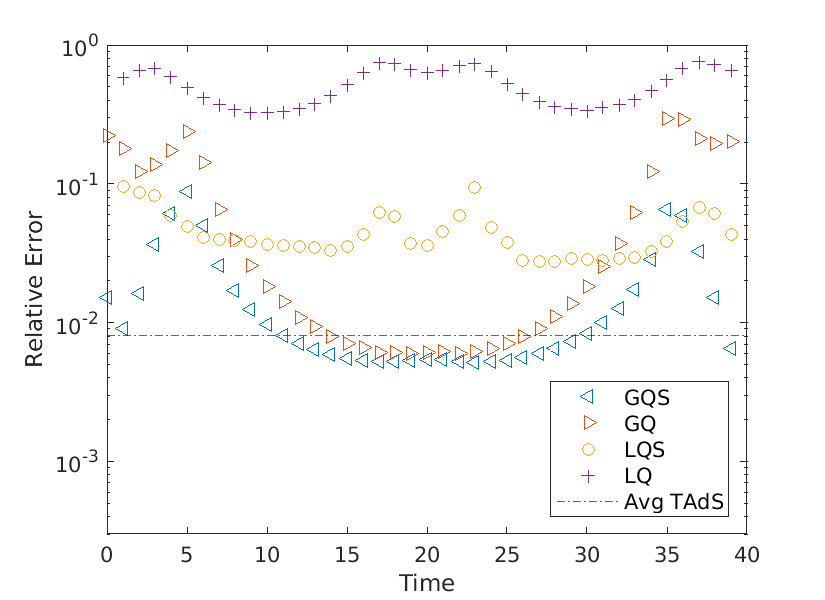}  
  \caption{Relative boundary error as a function of time for a free fermion system after a quench. GQS, GQ, LQS, LQ denote global quench with smoothing, global quench without smoothing, local quench with smoothing, and local quench without smoothing, respectively. Dashed line marks the average relative error for the thermal AdS entropy data.}
  \label{fig:quench_err_time}
\end{figure}

Finally, we show the relative errors corresponding to the superposition of thermal states in Figure~\ref{fig:superpose_err}. We find two regimes for the geometricality of the superposition, depending on the magnitude of the mixing term $G_NH(p)$. We find the superposition to be non-geometric when the entropy of mixing provides a significant correction to the geodesic lengths. This implies a generally non-geometrical construction when $G_NH(p)\approx 1$, as would be the case when the gravitational coupling $G_N$ is strong, or when the entropy of mixing is made large by superposing a large number of distinct geometries.

Conversely, in the weak coupling limit with few distinct superpositions, the contribution of the mixing term $G_NH(p)$ to the geodesic lengths becomes negligible, causing the entropy operator to be approximately linear in this regime. This leads to a smooth, geometric interpolation between the two distinct geometries as we tune the probability of mixing. The clear separation between the two cases is clearly illustrated in Figure~\ref{fig:superpose_err}, where we plot the relative errors for a state constructed as a mixture of two distinct thermal states.

\begin{figure}[H]
  \centering
  \includegraphics[width=0.6\linewidth]{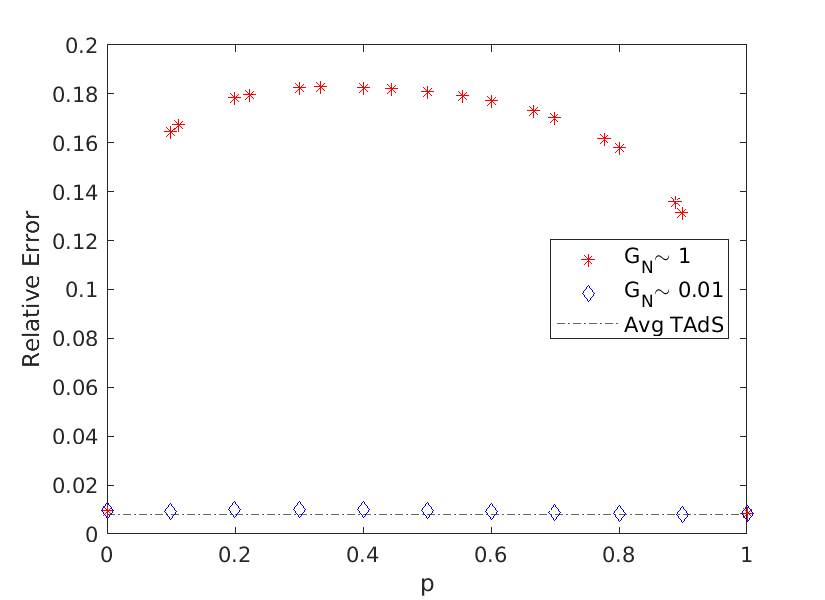}  
  \caption{Boundary relative errors for superpositions of thermal AdS geometries. }
  \label{fig:superpose_err}
\end{figure}

In summary, we find that the boundary relative error provides a useful measure for the extent to which a state is geometric or non-geometric. Through the discretized Radon transform, we confirm some existing expectations for the geometric nature of certain holographic states, as well as states arising from many-body systems that are generally believed to be non-geometric. We find that generic low-lying energy states of a free fermion system are indeed dominated by non-geometric contributions as indicated by our algorithm. However, there also exists states in these systems that have small relative error, such as the configurations following a global quench. It is suspected that the enhanced symmetries for the globally deformed states contributes to a more geometrical reconstruction. Moreover, quench dynamics reveal large-scale patterns of geometric evolution, despite the states themselves being non-geometric as characterized by the relative error. Further advances in the tensor Radon transform and its range characterization may help us understand what kind of properties in the boundary data lead to good reconstructions, and the nature of the large-scale behavior revealed by the numerical transform.

\section{Discussion}
\label{sec:discussion}

In this work, we took some first steps in addressing the explicit bulk metric reconstruction problem in holography, as well as the question of whether non-geometrical bulk states can be detected from boundary data alone. Motivated by the tensor Radon transform, we provided - and explicitly implemented - an algorithm that reconstructs the bulk metric tensor without any \textit{a priori} assumptions on the symmetries or form of the metric tensors, other than the fact that they are perturbatively close to AdS. We applied this reconstruction to entanglement data from holographic systems with large $N$, as well as to that of a $1d$ free fermion system. The reconstruction was also applied, time-slice by time-slice, to time-evolving states. We confirmed that, assuming the linearized Einstein's equations, thermalization following a global quench in holographic systems is consistent with a shell of in-falling matter in the bulk that finally settles into a state that appears to be gravitationally bound deep inside the bulk. This kind of behaviour is absent in free fermion systems where the corresponding shockwave of ``in-falling matter'' repeatedly oscillates between the boundary and the bulk due to the integrability of the system. 

We also provided a partial answer to whether a given state in the conformal field theory has a well-defined geometric dual on the gravity side. We find that boundary reconstruction errors provide a quantitative measure that distinguished ``geometrical'' states, such as the ones we find in large $N$ theories with semi-classical duals, from ``non-geometrical'' states like the generic excited states of a free fermion spin chain. This is a precise, albeit coarse, way of understanding what portion of the boundary data lies outside the range of the tensor Radon transform, and therefore cannot be interpreted as a tensor field on a hyperbolic background. In the instances we have examined, this measure was an effective indicator of non-geometricality.

Finally, our initial attempts in developing this algorithm indicate that efficient numerical reconstructions of the bulk metric tensor from entanglement entropy data, at least at the linearized level, are achievable. This is partly because the optimization problem we consider is linear in nature and can be solved in polynomial time with low computational power. As such, it provides an efficiently computable tomographic procedure that translates boundary entropy data, where the underlying spacetime and gravitational dynamics are hidden, to a setting where it is manifest. Although a wealth of previous analytic results are available, the flexibility of numerical studies is also invaluable, as we have found in various areas of physics from quantum many-body systems to numerical general relativity. As we start to move away from the more tractable dynamics of quantum systems that carry simplifying assumptions such as symmetries, the full force of numerical methods will prove to be extremely helpful. Here we provide one such preliminary construction which can hopefully provide a stepping stone for future advances.

Further work is clearly required to improve upon our first efforts. Here we list but a few major directions where progress can be made. 

On the front of new results in mathematics, especially related to tensor Radon transform:
\begin{enumerate}[label=\roman*]
    \item To put our geometry detector on a firmer mathematical footing, one needs a rigorous characterization of the range of tensor Radon transforms on curved backgrounds, in particular, the hyperbolic background. In the language of holographers, we are in need of a set of necessary\cite{Bao:2015bfa,Bao:2019bib} and sufficient conditions that characterizes what type of quantum states have semi-classical dual geometries. A complete characterization of the range of Radon transform on hyperbolic space, for example, precisely provides such conditions that checks whether a set of entanglement data is dual to a semi-classical geometry close to AdS. 
    
    \item It is also crucial to obtain an explicit reconstruction formula in closed form. Such formulae are known for the Euclidean background, and for the case of scalar or vector Radon transform on curved backgrounds. Development of these results will further enable calculations in the continuum limit, which is relevant for AdS/CFT. 
   \end{enumerate} 
   
On the front of new results in physics:
\begin{enumerate}[label=\roman*]

    \item We hope to extend this method to higher dimensions, which so far has faced the most obstacles in bulk metric reconstruction. Because area variation can already be cast as a tensor Radon transform in arbitrary dimensions \cite{Czech:2016tqr}, the idea of discretization followed by gauge fixing and linear optimization may simply be extended to minimal surfaces. Similarly, X-ray transforms in higher dimensions that makes use of correlation data can also be viable \cite{Porrati:2003na}. This may provide a numerically computable alternative to other methods \cite{Myers:2014jia,Balasubramanian:2018uus,Bao:2019bib}. 
    
    \item Although we examine dynamical cases, we only reconstruct the spatial metrics for each individual time slice. Reconstruction of the full linearized space-time metric perturbation against the AdS background should also be possible by gluing together these spatial slices. However, one has to be careful about gauge fixing and the role played by the time coordinate. More work is needed to clarify these constructions. 
    
    \item It is also worthwhile to go beyond linearized level. A known approach is to apply the inverse Radon transform iteratively, such that the background geometry is updated using the reconstructed metric perturbation $h_{ij}$. Such method is used in the geophysics community\cite{dahlen_tromp_1999,Ved}. However, we have to be careful about the change in extremal surface beyond the linear level, especially in the dynamical cases. This is worthwhile though, since working beyond the linearized level should allow a better understanding of non-trivial topologies, for example. As we have seen in this work, they are invisible at the linearized level despite our expectation that collapse of bulk matter should lead to such changes in geometry. 
    
    \item Reconstruction of entanglement data from other quantum systems. Thus far, we have only seen limited application of this reconstruction algorithm, largely limited by the availability of entanglement data. Nevertheless, it may be possible to compute these entanglement approximately for smaller but more interesting quantum many-body systems using tensor networks or other classical techniques. It may also be worthwhile to obtain data for geodesic lengths via other more accessible data, such as correlation functions or mutual information. 
    
    \item In light of the difficulty in computing von Neumann entropies for quantum systems, we can explore the possibility of using quantum Renyi entropies, which also have a geometric interpretation in holography \cite{Dong:2016fnf}. A prominent advantage is its computability using numerics and measurability in actual quantum systems. Reconstructions using such data can enable us to directly  acquire entropy data from quantum simulations and experimental setups. 
    
    \item Similar techniques using Radon transforms to recover the metric tensor of geometry emerged from entanglement are also applicable to near-flat manifolds outside the context of AdS/CFT. In fact, the Radon transform on flat-space is much better understood. Similar reconstructions should be possible for constructions like \cite{Cao:2017hrv} where, for instance, the relevant quantum states can be low energy states of a gapped local Hamiltonian.
    \end{enumerate}
 
   Finally, there is also much progress to be made in our reconstruction algorithm and related numerics. 
   \begin{enumerate}[label=\roman*]    
    \item One area of improvement is a data-driven modelling. For our current work, we fixed a particular tiling that is easy to implement. However, it has been shown that a tiling depending on the boundary data can lead to improvement for bulk reconstruction \cite{Lekic}. These include Voronoi partitions or improved modelling using Bayesian inference \cite{Tarantola}.
    
    \item Our discretization, constraint, and interpolation methods are all extremely simplistic, as appropriate for a proof of concept implementation. Numerous improvements can be made to improve the accuracy of the reconstruction algorithm, for example: proper triangular meshes, finite-element methods, better regularization techniques, etc. We expect that the reconstruction procedure here can be drastically improved in fidelity by adapting proper numerical techniques.
\end{enumerate}


\section*{Acknowledgements}
We thank Ping Gao, Vedran Lekic, Francois Monard, Chris Pattison, Gunther Uhlmann, Zhi-Cheng Yang, and Larry Zeng for helpful discussions.  C.C., X.L.Q., and B.G.S. acknowledge the support by the DOE Office of Science, Office of High Energy Physics, through the grant de-sc0019380. C.C. is also supported by the U.S. Department of Defense and NIST through the Hartree Postdoctoral Fellowship at QuICS. C.C. and B.G.S. acknowledge the support by the Simons Foundation as part of the It From Qubit Collaboration. ET acknowledge funding provided by the Institute for Quantum Information and Matter, an NSF Physics Frontiers Center (NSF Grant
PHY-1733907), the Simons Foundation It from Qubit Collaboration, the DOE QuantISED
program (DE-SC0018407), and the Air Force Office of Scientific Research (FA9550-19-1-
0360). ET also acknowledges the support of the Natural Sciences and Engineering Research
Council of Canada (NSERC).

\appendix
\section{The Tensor Radon Transform}
\label{app:TRT}

Here, let us formally define the geodesic tensor Radon transform. We begin with an introduction to the tensor Radon transform in general, but will quickly specialize to the special case of the $2$-tensor Radon transform on the Poincare disk (with finite cutoff). 

In short, the $m$-tensor geodesic Radon transform $R_m$ is a map which takes a symmetric $m$-tensor field defined on a sufficiently well-behaved Riemannian manifold $M$ (with boundary) to the space of geodesics on that manifold.

\subsection{General Definitions}

Let $(M,g)$ be an $n$-dimensional Riemmanian manifold with boundary $\partial M$ and metric $g$. We say that $(M,g)$ is a \emph{simple} manifold if $\partial M$ is strictly convex\footnote{Namely, given any two points $p,q \in \partial M$, there exists a geodesic segment connecting $p$ and $q$ that meets $\partial M$ only at $p$ and $q$.} and any two points in $M$ are connected by a unique geodesic segment which depends smoothly on the endpoints~\cite{pestov-uhlmann2003}. Alternatively, a simple manifold is one in which the boundary is strictly convex, and where the exponential map $\exp_p:\exp_p^{-1}(M)\rightarrow M$ is a diffeomorphism for every $p\in M$. Fixing some $p\in M$, we may identify a simple manifold with a strictly convex domain $\Omega$ of $\mathbb{R}^n$.

The simplicity of a manifold is a sufficient condition for the geodesic Radon transform to be well-defined~\cite{tomography_survey}. There are more general conditions available, but simplicity will generally be sufficient for our purposes. From now on, unless otherwise stated, all of our manifolds will be assumed to be simple.

Let $SM$ denote the unit circle bundle of $M$. The bundle $SM$ is the collection of all pairs $(p,v)$, where $p\in M$ and where $v \in T_pM$ is a unit tangent vector at $p$. The boundary of the unit circle bundle, denoted $\partial SM$, consists of all such pairs where $p \in \partial M$. The boundary of the unit circle bundle naturally splits into two components, $\partial_{+}SM$ consisting of all the inward pointing vectors, and $\partial_{-}SM$ consisting of all the outward pointing vectors. We will define both components to be closed, i.e., vectors tangent to $\partial M$ will be in both $\partial_{+}SM$ and $\partial_{-}SM$. 

Given $(p,v)\in \partial_{+}SM$, let $\gamma_{p,v}:[0,\tau(p,v)] \rightarrow M$ denote the unique unit speed geodesic through $(p,v)$, i.e., the unique geodesic such that
\begin{align}
\gamma_{p,v}(0) = p,\qquad\text{and}\qquad \dot{\gamma}_{p,v}(0)=v.
\end{align}
The parameter $\tau(p,v)$ denotes the \emph{exit time}\footnote{Or since $\gamma_{p,v}$ is unit speed parametrized, the total arclength of $\gamma_{p,v}$.} of $\gamma_{p,v}$, i.e., the first non-zero time such that $\gamma_{p,v}(\tau) \in \partial M$. Note the exit time is well-defined under the assumption that the underlying manifold is simple.

Now, let $f_{i_1\cdots i_m}$ be a smooth, symmetric (covariant) $m$-tensor field on $M$. Then the Radon transform of $f$ is defined by
\begin{align}
    R_m[f](p,v) = \int_0^{\tau(p,v)}f_{i_1\cdots i_m}(\gamma_{p,v}(s))\,\dot{\gamma}_{p,v}^{i_1}(s)\cdots \dot{\gamma}_{p,v}^{i_m}(s)\ ds.
\end{align}
Thus, the Radon transform is a map $R_m:S_m(M) \rightarrow C^\infty(\partial_{+}SM)$ which takes the space $S_m(M)$ of smooth, symmetric (covariant) $m$-tensors on $M$ to the space $C^\infty(\partial_{+}SM)$ of smooth functions on the inward pointing boundary unit circle bundle component $\partial_{+}SM$. Since we can uniquely identify each $(p,v)\in \partial_+ SM$ with a corresponding unit speed geodesic $\gamma_{p,v}$, it will be often convenient to consider the Radon transform as a map on the space of boundary anchored geodesics. 

\subsection{$s$-Injectivity}

There are a few natural questions we may ask for the Radon transform, the foremost being the surjectivity and the injectivity of the transform. We will not comment much on the range of the tensor Radon transform, except to note that the tensor Radon transform is generally not surjective. This, of course, corresponds to the well known fact that not all boundary states have a well-defined bulk dual. A useful analytic characterization of the range remains an open problem for the tensor Radon transform on generic manifolds.

Let us now consider the problem of injectivity. The tensor Radon transform has a natural kernel. Let $(M,g)$ be an $n$-dimensional simple Riemannian manifold. We will let $dV$ denote the canonical volume form, which is locally given by 
\begin{align}
    dV = \sqrt{|g|}\, dx^1\wedge \cdots \wedge dx^n,
\end{align}
where $(x^1,\cdots ,x^n)$ is some oriented chart, and where $|g|$ is the determinant of the metric $g_{ij}$ in that chart. We will let $\nabla$ denote the Levi-Civita connection on $M$. Now, let $S_mM$ denote the space of (covariant) symmetric $m$-tensors on $M$. We will define the \emph{inner derivative}\footnote{Despite the confusingly similar notation, $\mathbf{d}$ is \emph{not} the exterior derivative $d$, which acts on forms and not symmetric tensors. Unfortunately, the notation is somewhat well-established in the integral geometry community, and so we will stick to it. We will never need to use the exterior derivative in this paper, but nevertheless we will denote the inner derivative with boldface font $\mathbf{d}$ as a reminder that it is not the exterior derivative.} $\mathbf{d}:S_mM \rightarrow S_{m+1}M$ by
\begin{align}
    \mathbf{d}f = \sigma\nabla f,
\end{align}
where $\sigma$ denotes complete symmetrization. In local coordinates, we simply have
\begin{align}
    (\mathbf{d}f)_{i_1\cdots i_{m+1}} &= f_{(i_1 \cdots i_m ; i_{m+1})},
\end{align}
where parentheses indicate the complete symmetrization of the contained indices as usual. We will likewise define the \emph{divergence}\footnote{Again, $\boldsymbol{\delta}$ is not the co-exterior derivative. We will use the same boldface convention.} $\boldsymbol{\delta}:S_mM \rightarrow S_{m-1}M$ to be
\begin{align}
    \boldsymbol{\delta}f = \mathrm{Tr}_{m,m+1}(\nabla f),
\end{align}
where $\mathrm{Tr}_{m,m+1}$ denotes (Riemannian) contraction between the $m$th and $(m+1)$th arguments. In local coordinates, we have
\begin{align}
    (\boldsymbol{\delta}f)_{i_1,\cdots,i_{m-1}} = f_{i_1,\cdots ,i_{m-1}, j;k}g^{jk}.
\end{align}

The operators $\mathbf{d}$ and $-\boldsymbol{\delta}$ are adjoint for compactly supported symmetric tensor fields on $M$. More generally, for any compact region $D\subseteq M$, we have
\begin{align}
    \int_D \langle \mathbf{d}u,v \rangle + \langle u,\boldsymbol{\delta} v \rangle\ dV = \int_{\partial D} \langle i_{\nu}u,v \rangle\ dS,
\end{align}
where $u$ and $v$ are (sufficiently smooth) symmetric tensors of the appropriate orders, $i_{\nu}$ denotes interior multiplication with respect to an outward pointing normal $\nu$, and where $dS$ is the induced volume form on $\partial D$. The inner product $\langle\cdot ,\cdot\rangle$ on $S_m M$ is given by complete contraction, i.e.,
\begin{align}
    \langle u, v\rangle = g^{i_1j_1}\cdots g^{i_mj_m}u_{i_1,\cdots, i_m}v_{j_1,\cdots ,j_m}.
\end{align}

The significance of the operators $\boldsymbol{\delta}$ and $\mathbf{d}$ are as follows. Let $H^k(S_mM)$ denote the Sobolev space of $m$-symmetric tensors, i.e., the space of all sections which are $k$-times (weakly) differentiable, and such that all derivatives are locally square integrable. Each $H^k(S_mM)$ can be given the structure of a Hilbert space when $M$ is compact. Then we have the generalized Helmholtz decomposition as follows:

\begin{theorem}[Generalized Helmholtz Decomposition~\cite{sharafutdinov1994integral}]\label{thm:helmholtz}
Let $(M,g)$ be a compact Riemannian manifold with boundary. Let $k\ge 1$ and $m\ge 0$ be integers. Given any section $f\in H^k(S_mM)$, there exists uniquely determined $f^s \in H^k(S_mM)$ and $v \in H^{k+1}(S_{m-1}M)$ such that
\begin{align}
    f = f^s + \boldsymbol{d}v,\qquad\text{where}\quad \boldsymbol{\delta}f^s = 0,\quad\text{and}\quad v|_{\partial M}=0.
\end{align}
The fields $f^s$ and $\boldsymbol{d}v$ are called the \emph{solenoidal} and $\emph{potential}$ parts of $f$.
\end{theorem}

Note that in the case $m=1$, Theorem~\ref{thm:helmholtz} is just the usual Helmholtz decomposition
\begin{align}
    \mathbf{F} = -\nabla \varphi + \nabla \times \mathbf{A}
\end{align}
after identifying vectors and covectors using the metric.

The decomposition of a symmetric tensor field into solenoidal and potential parts gives us a natural identification of the kernel for the Radon transform. Indeed, we have the following result:

\begin{theorem}\label{thm:kernel}
Let $(M,g)$ be a simple manifold and let $R_m$ be the $m$-tensor Radon transform on $M$. Let $v \in H^{k+1}(S_{m-1}M)$ be a vector field such that $v|_{\partial M} = 0$. Then
\begin{align}
    R_m[\boldsymbol{d}v] = 0.
\end{align}
\end{theorem}

Therefore, Theorem~\ref{thm:kernel} identifies a natural kernel for the Radon transform, namely the space of all potential tensor fields. Since every sufficiently smooth tensor field can be decomposed uniquely into a potential and a solenoidal part, a natural question is whether the solenoidal part is uniquely recoverable from the Radon transform, i.e., whether the space of potential tensor fields exhausts the kernel of the Radon transform. If this is indeed the case, i.e., if $R_m[f] = 0$ implies $f^{\mathrm{s}} = 0$, then we say that the Radon transform is $s$-\emph{injective}.

The question of the $s$-inectivity of the Radon transform is a fundamental problem in integral geometry. The general case remains open, although the case for $2$-dimensional simple manifolds was settled in the affirmative~\cite{uhlmann-injective}:

\begin{theorem}[$s$-injectivity~\cite{uhlmann-injective}]
Let $(M,g)$ be a simple $2$-dimensional Riemannian manifold. Then the tensor Radon transform $R_m$ on $M$ is $s$-injective for all $m\ge 0$.
\end{theorem}

Note that in the case of the scalar transform for $m=0$, all scalar functions are automatically solenoidal, so the scalar Radon transform $R_0$ is injective in the usual sense.

Given the $s$-injectivity of the Radon transform, we can recover the bulk tensor field by imposing the \emph{solenoidal gauge condition}
\begin{align}
    f = f^{\mathrm{s}}.
\end{align}
The solenoidal gauge is the most commonly employed gauge condition for the tensor Radon transform, but it comes with some inconvenient features. In particular, it does not respect the decomposition of a tensor into its trace and traceless parts. The $2$-tensor Radon transform on a purely trace bulk tensor field is identical to the scalar Radon transform on the corresponding trace function, but the inversion of the $2$-tensor Radon transform under the solenoidal gauge introduces extraneous gauge degrees of freedom which causes the recovered field to disagree with the original. This is rather undesirable, since the corresponding scalar transform is purely injective and admits a unique recovery. We will instead employ an alternative gauge condition, which we introduce in section~\ref{app:gaugefix}, that makes the tensor Radon transform consistent with the scalar Radon transform for pure trace fields.

\section{The Radon transform on the Poincare disk}
\label{app:TRTpoincare}

Let us now focus our attention to the case of a single time-slice of $\mathrm{AdS}_3$. Such a time-slice is isomorphic to the hyperbolic plane, which we will consider in the Poincare disk model. Concretely, the Poincare disk is the Riemannian manifold of constant scalar curvature $R=-1$, defined within the open unit disk 
\begin{align}
    \bbD = \left\{z\in \bbC \mid |z|<1\right\}.
\end{align}
We will cover the Poincare disk using its natural Cartesian coordinates $(x,y)$, or equivalently, using complex coordinates $z=x+iy$. The metric is then given by
\begin{align}
    g(z) = \frac{4(dx^2 + dy^2)}{(1-x^2-y^2)^2}=4(1-|z|^2)^{-2}\,|dz|^2,
    \label{eqn:poincaremetric}
\end{align}
where $|dz|^2 = dz\cdot d\overline{z} = dx^2 + dy^2$. We will denote the Poincare disk by $\bbH = (\bbD, g)$. Geometrically, the geodesics of the Poincare disk are circular segments which are orthogonal to the boundary circle $S^1$.

Both for physical reasons, and to properly define the Radon transform, we cannot work with the Poincare disk in its entirety.\footnote{Technically, the Radon transform doesn't really care about the fact that the Poincare disk is non-compact, especially since there is still a well-defined notion of (asymptotic) boundary. However, it is much more convenient to work with a compact manifold with an actual boundary.} Rather, we must impose a cutoff. We will let $\kappa\in (0,1)$ be a cutoff radius, and consider the cutoff disk
\begin{align}
    \bbD_{\kappa} = \{z \in \bbC \mid |z| < \kappa\}.
\end{align}
Then we consider the cutoff Poincare disk to be the manifold defined by $\bbH_\kappa = (\bbD_\kappa, g)$.

We will be mainly interested in the Radon transform for symmetric $2$-tensors. The Poincare disk (with cutoff) is a simple manifold. As such, the Radon transform is always well-defined on the Poincare disk. To that end, let $f_{ij}$ be a symmetric $2$-tensor field on $\bbH_\kappa$. The Radon transform is then given by
\begin{align}
    R_2[h](p,v) = \int_{0}^{\tau(p,v)}f_{ij}(\gamma(s))\dot{\gamma}_{p,v}^i(s)\dot{\gamma}_{p,v}^j(s)\ ds.
\end{align}
Note that since the boundary of the (cutoff) Poincare disk is a circle, we may uniquely identify each geodesic $\gamma_{p,v}$ with a connected subregion $A$ (i.e., a circular arc) of the boundary. The corresponding geodesic is then the minimal surface with respect to that boundary subregion. In this way, we may regard the image of the Radon transform as a function on the space of boundary subregions.

The metric on the Poincare disk is an isotropic metric. Let us write 
\begin{align}
g(z) = 4(1-|z|^2)^{-2}|dz|^2 = e^{2\lambda(z)}|dz|^2.
\end{align}
We may then write a unit tangent vector as $v = e^{-\lambda}(\cos\theta, \sin\theta)$, where $\theta$ is the natural angular coordinate on the unit disk. Then we can decompose the Radon transform as
\begin{align}
    f_{ij}\dot{\gamma}^i\dot{\gamma}^j &= e^{-2\lambda}\left(f_{11}\cos^2\theta + 2f_{12}\cos\theta\sin\theta + f_{22}\sin^2\theta\right) \label{eq:radon_function}\\
    &= e^{-2\lambda}\left(\frac{f_{11}+f_{22}}{2} + f_{12}\sin2\theta + \frac{f_{11} - f_{22}}{2}\cos2\theta \right)\nonumber\\
    &= e^{-2\lambda}\left(h_0 + he^{2i\theta} + \overline{h}e^{-2i\theta}\right),\nonumber
\end{align}
where in the last line we've defined the components
\begin{align}
    h_0 &= \frac{1}{2}\tr{h},\\
    h &= \frac{1}{4}(f_{11} - 2if_{12} - f_{22}),\\
    \overline{h} &= \frac{1}{4}(f_{11} + 2if_{12} - f_{22}).
\end{align}
The components $h_0$, $h$, and $\overline{h}$ will be called the \emph{trace}, \emph{holomorphic}, and \emph{anti-holomorphic} parts of $f$, respectively. Note that in this way, we can always identify any symmetric $2$-tensor $f_{ij}$ with a function $\tilde{f} \equiv f_{ij}\dot{\gamma}^i\dot{\gamma}^j$ of harmonic content $(-2,0,2)$.

We can then equivalently write the tensor Radon transform in the $h$-components as
\begin{align}
    R_2[g](p,v) = \int_0^{\tau(p,v)}e^{-2\lambda(\gamma(s))}\left(h_0(\gamma(s)) + h(\gamma(s))e^{2i\theta(s)} + \overline{h}(\gamma(s))e^{-2i\theta(s)}\right)\ ds.
    \label{eqn:TRTholo}
\end{align}
Note that the trace component of the transform simply gives the scalar transform (after absorbing the metric into the definition of $h_0$). Therefore the tensor Radon transform $R_2$ is essentially equivalent to the scalar Radon transform for tensors which are purely trace.

As previously mentioned, the solenoidal gauge does not respect the decomposition of the tensor into trace and traceless parts. It will therefore be convenient to find a gauge which preserves this decomposition. We do so in Appendix~\ref{app:gaugefix}.

\section{The Holomorphic Gauge}
\label{app:gaugefix}
In this section, we present an alternative to the solenoidal gauge, which we call the \emph{holomorphic gauge}, which both uniquely fixes a solution space to the inverse Radon transform and preserves the the scalar part of the transform.

To define the holomorphic gauge, we first introduce some background. We employ global isothermal coordinates $(x,y,\theta)$ on the unit circle bundle $S\bbH$ of the Poincare disk, where $(x,y)$ are the usual Poincare coordinates on the base manifold, and where $\theta$ is a fiber coordinate indicating the angular direction of a tangent vector.

We can define the geodesic flow $X$, given in global isothermal coordinates on the unit circle bundle by
\begin{align}
    X = e^{-\lambda}\left(\cos\theta\,\partial_x + \sin\theta\,\partial_y + \left(-\sin\theta\,\partial_x\lambda + \cos\theta\,\partial_y\lambda\right)\partial_\theta\right),
\end{align}
where we write the metric as $g=e^{2\lambda}(dx^2+dy^2)$. Writing the vector field in terms of complex coordinates, we can decompose the geodesic flow as $X = \eta_+ + \eta_-$, where
\begin{align}
    \eta_+ &= e^{-\lambda}e^{i\theta}\left(\partial + i\partial\lambda\,\partial_\theta \right),\\
    \eta_- &= e^{-\lambda}e^{i\theta}\left(\overline{\partial} + i\overline{\partial}\lambda\,\partial_\theta \right) = \overline{\eta}_+,
\end{align}
where $\partial = \frac{1}{2}(\partial_x - i\partial_y)$ and $\overline{\partial} = \frac{1}{2}(\partial_x + i\partial_y)$ are the Wirtinger derivatives. 

Let $\Omega_k = L^2(S\bbH) \cap \ker(\partial_\theta - ik)$ be the space of square-integrable functions on the unit circle bundle with fixed harmonic content $k$. Then it can be shown~\cite{Monard_2014} that $\eta_{\pm}$ are smooth elliptic differential operators such that $\eta_{\pm}:\Omega_k \rightarrow \Omega_{k\pm 1}$ for any $k\in \bbZ$. In particular, note that $\Delta \equiv \eta_+\eta_- = \eta_-\eta_+$ is also a smooth elliptic partial differential operator.

Note that the geodesic flow $X$ is naturally related to the Radon transform as follows~\cite{sharafutdinov1994integral}: Let $f_{ij}$ be a symmetric $2$-tensor field, and let us define the function
\begin{align}
u\left(x,y,\theta\right)=\int_{0}^{\tau_{+}(\gamma)}f_{ij}\dot{\gamma}^{i}\dot{\gamma}^{j}\ ds,
\end{align}
where $\gamma$ denotes the unique unit speed geodesic through $\left(x,y\right)$, with initial angle $\theta$, and where $\tau_{+}(\gamma)$ denotes the exit time of the geodesic. Note that we have $u|_{\partial_{-}SM}=0$ and 
\begin{align}
u|_{\partial_{+}SM}=R_{2}[f],
\end{align}
by construction. If $\gamma_{(x_{0},y_{0},\theta_{0})}$ is a geodesic through $\left(x_{0},y_{0},\theta_{0}\right)$, then 
\begin{align}
u\left(\gamma_{(x_{0},y_{0},\theta_{0})}(t),\theta_{\left(x_{0},y_{0},\theta_{0}\right)}(t)\right)=\int_{t}^{\tau_{+}(\gamma)}f_{ij}\dot{\gamma}^{i}\dot{\gamma}^{j}\ ds,
\end{align}
where $\theta_{\left(x_{0},y_{0},\theta_{0}\right)}(t)\equiv\arg\left[\dot{\gamma}_{\left(x_{0},y_{0},\theta_{0}\right)}(t)\right]$ denotes the angle of the tangent vector to $\gamma_{\left(x_{0},y_{0},\theta_{0}\right)}$ at time $t$. Differentiating with respect to $t$, we get
\begin{align}
\dot{\gamma}^{i}\frac{\partial u}{\partial x^{i}}+\dot{\theta}\frac{\partial u}{\partial\theta}=-f_{ij}\dot{\gamma}^{i}\dot{\gamma}^{j}.\label{eq:transport}
\end{align}
Note that the left-hand side is precisely the expression $Xu$.\footnote{Let us explicitly calculate $\dot{\theta}$ here. We start with the geodesic equation \begin{align}
\ddot{\gamma}^{i}+\Gamma_{\ jk}^{i}\dot{\gamma}^{j}\dot{\gamma}^{k}=0.
\end{align}
Evaluating the $x$ component, for example, we get
\begin{align}
\frac{d}{ds}(e^{-\lambda}\cos\theta)+e^{-2\lambda}\partial_{1}\lambda(\cos^{2}\theta-\sin^{2}\theta)+2e^{-2\lambda}\partial_{2}\lambda\sin\theta\cos\theta=0.
\end{align}
Taking the derivative and expanding, we then have
\begin{align}
-e^{-2\lambda}(\partial_{1}\lambda\cos^{2}\theta+\partial_{2}\lambda\sin\theta\cos\theta)-e^{-\lambda}\sin\theta\cdot\dot{\theta}+e^{-2\lambda}\partial_{1}\lambda(\cos^{2}\theta-\sin^{2}\theta)+2e^{-2\lambda}\partial_{2}\lambda\sin\theta\cos\theta=0,
\end{align}
which simplifies as
\begin{align}
-e^{-\lambda}\partial_{1}\lambda\sin\theta+e^{-\lambda}\partial_{2}\lambda\cos\theta=\dot{\theta},
\end{align}
so that equation~\eqref{eq:transport} is given by
\begin{align}
e^{-\lambda}\left[\cos\theta\frac{\partial}{\partial x}+e^{-\lambda}\sin\theta\frac{\partial}{\partial y}+\left(-\sin\theta\frac{\partial\lambda}{\partial x}+\cos\theta\frac{\partial\lambda}{\partial y}\right)\frac{\partial}{\partial\theta}\right]u\left(x,y,\theta\right)=Xu\left(x,y,\theta\right)=-\tilde{f}\left(x,y,\theta\right),
\end{align}
where $\tilde{f}\equiv f_{ij}\dot{\gamma}^{i}\dot{\gamma}^{j}$, as given in equation~\eqref{eq:radon_function}.} Denoting by $u^f$ the unique solution to the \emph{transport equation}~\eqref{eq:transport} with boundary condition $u^f|_{\partial_{-} SM} = 0$, it follows that the Radon transform $R_2[f]$ is given by
\begin{align}
    R_2[f] = u^f|_{\partial_+ SM}.
\end{align}
The transport equation can therefore be seen as the differential form of the Radon transform.

The key result leading to the holomorphic gauge is then the following theorem.

\begin{theorem}\label{thm:holomorphic_gauge}
For any symmetric $2$-tensor $f \in L^2(SM)$, there exists a unique $2$-tensor $h \in L^2(SM)$ such that $R_2f = R_2h$, and such that $h$ is of the form
\begin{align}
    h = h_0 + h_{2} + h_{-2},
\end{align}
where $h_0 \in L^2(M)\cap \Omega_0$, and where $h_{\pm 2} \in \ker\eta_{\mp}\cap \Omega_{\pm 2}$.
\end{theorem}
\begin{proof}
We adapt the proof from Theorem 1 of~\cite{Monard2015}, which covers the case where the underlying metric is the usual flat Euclidean metric. Let $f\in L^2(SM) \cap \Omega_k$ be given. Then consider the differential equation
\begin{align}
    \eta_- f = \Delta v,\\
    v|_{\partial SM} = 0,
\end{align}
where $v \in H^1(SM)\cap \Omega_{k-1}$. Since $\Delta$ is a smooth elliptic operator, it follows from the standard theory of elliptic differential equations that the above system admits a unique (weak) solution $v$. Given such a solution, let us define $g = f - \eta_+ v$. It follows that 
\begin{align}
\eta_- g = \eta_-(f- \eta_+ v) = 0.
\end{align}
This shows that each $f\in L^2(SM)\cap \Omega_k$ can be written uniquely as
\begin{align}
    f = \eta_+ v + g,
\end{align}
where $v \in H^1(SM)\cap \Omega_{k-1}$, $v|_{\partial SM} = 0$, and where $g\in L^2(SM)$ such that $\eta_- g = 0$.

Next, let $f$ be a given $2$-tensor which we may write as
\begin{align}
    f = f_0 + f_{-2}+ f_2,
\end{align}
where $f_k \in \Omega_k$. The transport equation gives us 
\begin{align}
    Xu = -f,\\
    u|_{\partial_- SM} = 0,\\
    u|_{\partial_+ SM} = R_2f.
\end{align}
Let us apply the previously derived decomposition to write $f_2 = \eta_+ v_1 + g_2$, where $v_1|_{\partial SM} = 0$ and where $\eta_- g_2 = 0$. This gives us
\begin{align}
    f_2 &= \eta_+ v_1 + g_2\\
        &= Xv_1 -\eta_- v_1 + g_2,
\end{align}
which allows us to write the transport equation as
\begin{align}
    X(u + v_1) = -(f_0 - \eta_- v_1 + f_{-2} +  g_2),
\end{align}
where $\eta_- g_2 = 0$. Since $v_1|_{SM} = 0$, it follows that we have $(u+v_1)|_{\partial_+ SM} = u|_{\partial_+ M} = R_2f$. If we define the $2$-tensor $h$ by
\begin{align}
    h = h_0 + h_{-2} + h_2,
\end{align}
where $h_0 = f_0 - \eta_- v_1$ and $h_{-2} = f_{-2}$, then $h$ satisfies $R_2h = R_2f$, and is such that $\eta_- h_2 = 0$. We may repeat this reasoning with the $f_{-2}$ term  (using the complex conjugate of the previous decomposition) to obtain the desired result.
\end{proof}

\begin{definition}
We will define a $2$-tensor whose components satisfy the conditions of Theorem~\ref{thm:holomorphic_gauge} to be in the \emph{holomorphic gauge}. 
\end{definition}

Importantly, let us note that if a $2$-tensor $f$ is purely scalar, i.e., $f = f_0$, then it is trivially already in the holomorphic gauge. Thus the holomorphic gauge is a gauge which respects the scalar part of the transform. This is to be contrasted with the solenoidal gauge, which will introduce spurious off-diagonal components even for scalar tensor fields.

Since $h_{\pm 2} \in \Omega_{\pm 2}$, let us write $h_{2} = he^{2i\theta}$ and $h_{-2} = \overline{h}e^{-2i\theta}$, where $h,\overline{h} \in L^2(M)$. In this notation, the holomorphic gauge condition reads
\begin{align}
    \eta_+(\overline{h}e^{-2i\theta}) = 0,\\
    \overline{\eta_+(\overline{h}e^{-2i\theta})} = 0.
\end{align}
We have
\begin{align}
    \eta(\overline{h}e^{-2i\theta}) &= e^{-\lambda}e^{i\theta}\left(\partial + i\partial\lambda\,\partial_\theta \right)(\overline{h}e^{-2i\theta})\label{eq:hologauge1}\\
      &= e^{-\lambda}e^{-i\theta}\left(\partial + 2\partial\lambda \right)h,\label{eq:hologauge2}
\end{align}
so the holomorphic gauge conditions simplify to the Schrodinger type equations
\begin{align}
    (\partial + 2\partial \lambda)\overline{h} = 0,\\
    (\overline{\partial} + 2\overline{\partial} \lambda)h = 0.
    \label{eqn:hologaugeconstraint}
\end{align}
We can solve these equations by introducing an integrating factor of $e^{2\lambda}$. Then the equation for $h$ gives us
\begin{align}
    0 &= e^{2\lambda}\overline{\partial}h + 2\overline{\partial} \lambda\,e^{2\lambda}h\\
    &= \overline{\partial}(e^{2\lambda}h).
\end{align}
This amounts to saying that the solutions are holomorphic functions (up to an exponential factor) on the unit disk $\bbD$. We can therefore generate any solution as follows: Given any function on the unit circle $h|_{\partial\bbD}:\partial\bbD \rightarrow \bbC$, we can extend $h$ into the unit disk using Cauchy's integral formula to get
\begin{align}
    h(z) = \frac{c_\lambda e^{-2\lambda(z)}}{2\pi i}\oint_{\partial\bbD}\frac{h|_{\partial \bbD}(w)}{w-z}\ dw,\label{eq:holo_gauge_solution}
\end{align}
where we write $c_\lambda = e^{2\lambda(1)}$ to denote the (constant) value of $e^{2\lambda}$ on the unit circle. We will use this convenient reconstruction property of solutions in the holomorphic gauge to benchmark the numerical (inverse) Radon transform in Appendix~\ref{app:benchmark}.

\section{The Numerical (Inverse) Radon Transform}
\label{app:numrecon}

In this appendix, we describe the details of the numerical (inverse) Radon transform.

The manifold we work with is the hyperbolic plane $\bbH_2$, which is modeled as the Poincare disk, i.e., the unit disk equipped with the canonical hyperbolic metric
\begin{align}
    g = \frac{4}{(1-x^2-y^2)^2}(dx^2+dy^2),
\end{align}
where $(x,y)$ are global Poincare coordinates. For physical reasons, and to make the Radon transform well-defined, we must impose a cutoff on the Poincare disk. We thus pick a constant $\kappa \in (0,1)$ and work with the Poincare disk restricted to $r\le \kappa$. Equivalently, the cutoff Poincare disk can be regarded (after a rescaling of the metric) as a model for a hyperbolic plane with curvature equal to $-\kappa^2$.

To perform the numerical Radon transform, we must first discretize the Poincare disk. A natural first choice would be to use a uniform tiling, a choice which conforms best to the intrinsic symmetries of the Poincare disk. However, the Gauss-Bonnet theorem places limitations on how fine a uniform tiling can be, and the inability to take the tile size to zero is an unwieldy restriction. Instead, we opt for a simple square tessellation which we perform in the Beltrami-Klein model.

The Beltrami-Klein model is related to the Poincare disk model through the change of coordinates
\begin{align}
    (r,\theta) \mapsto (R,\Theta) = \left(\frac{2r}{1+r^2},\theta\right), 
\end{align}
where $(r,\theta)$ denotes polar Poincare coordinates, and $(R,\Theta)$ denotes polar Beltrami-Klein coordinates. The Beltrami-Klein model has the convenient property that geodesics are straight lines. This makes a square tessellation in the Beltrami-Klein model the closest analogue to a regular Euclidean square tessellation for the hyperbolic plane. The figures present throughout this paper showcase the corresponding square tessellation in the Poincare disk.

Given a tessellation of the Poincare disk, we can discretize any function $f$ by assigning to a given tile $\cT$ the value of $f$ at the centroid of $\cT$. Ordering the tiles arbitrarily, we can regard the discretized functions as vectors,
\begin{align}
    \mathbf{f} = \begin{pmatrix}f(\cT_1) \\ \vdots \\ f(\cT_N)\end{pmatrix}.
\end{align}

We must also discretize the (ideal) boundary of the Poincare disk for the Radon transform. We do so by placing $M$ equally spaced boundary sites on the unit circle. We will then consider the collection of all geodesics originating from one boundary site and terminating on another. We can order the boundary sites by their angular position on the unit circle, and the geodesics lexicographically by the angular positions of their endpoints.

For each geodesic of the background geometry, the integral \eqref{eqn:TRTholo} can then be discretized by replacing the functions $h_0, h$ and $h$ with their piece-wise constant discretizations: 
\begin{align}
I_2[g](\gamma_j) &= \int_{\gamma_j}e^{-2\lambda(\gamma_j(s))}\left(h_0(\gamma_j(s)) + h(\gamma_j(s))e^{2i\theta(s)} + \overline{h}(\gamma_j(s))e^{-2i\theta(s)}\right)\ ds \\
&\approx \sum_{\mathcal{T}}\left[ h_0(\mathcal{T}) W_0(\gamma_j, \mathcal{T}) + h(\mathcal{T})W(\gamma_j, \mathcal{T})+\overline{h}(\mathcal{T})\overline{W}(\gamma_j, \mathcal{T})\right],
\end{align}
where $W_0, W$, and $\overline{W}$ contain the information on the geodesic and the remaining parts of the integrand. Explicitly, $W_0$ and $W$ are defined by
\begin{align}
 W_0(\gamma_j,\mathcal{T}) &= \int_{\gamma_j} e^{-2\lambda(\gamma_j(s))}\,\mathbf{1}_{\cT}(\gamma_j(s))\ ds,\\
W(\gamma_j,\mathcal{T}) &= \int_{\gamma_j} e^{-2\lambda(\gamma_j(s))+2i\theta(s)}\,\mathbf{1}_{\cT}(\gamma_j(s))\ ds,
\end{align}
where $\mathbf{1}_{\cT}(\gamma(s))$ is an indicator function such that
\begin{equation}
\mathbf{1}_{\cT}(x) = \begin{cases}1 & x \in \cT,\\
0 & x \notin \cT.
\end{cases}
\end{equation}
Note that $W_0$ is nothing but the arc lengths of the geodesic segments that intersects a tile $\cT$. The function $W$, on the other hand, also contains a complex weight which captures the directionality of the geodesic in the tile.

It's important to note that $W_0$ and $W$ depend only on the particular choice of discretization, and can be pre-computed.

We can collect all of the quantities into a matrix equation. Let $\mathbf{W}$ be the $K \times 3N$, where $K=\binom{M}{2}$, defined by

\begin{align}
\tiny{
    \mathbf{W} = \begin{pmatrix}W_0(\gamma_1,\cT_1) &  \cdots & W_0(\gamma_1,\cT_N) & W(\gamma_1, \cT_1) & \cdots & W(\gamma_1,\cT_N) & \overline{W}(\gamma_1,\cT_1) & \cdots & \overline{W}(\gamma_1,\cT_N) \\
    W_0(\gamma_2,\cT_1) &  \cdots & W_0(\gamma_2,\cT_N) & W(\gamma_2, \cT_1) & \cdots & W(\gamma_2,\cT_N) & \overline{W}(\gamma_2,\cT_1) & \cdots & \overline{W}(\gamma_2,\cT_N)\\
    \vdots & \ddots & \vdots & \vdots & \ddots & \vdots  & \vdots & \ddots & \vdots\\
    W_0(\gamma_{K},\cT_1) &  \cdots & W_0(\gamma_{K},\cT_N) & W(\gamma_{K}, \cT_1) & \cdots & W(\gamma_{K},\cT_N) & \overline{W}(\gamma_{K},\cT_1) & \cdots & \overline{W}(\gamma_{K},\cT_N)
    \end{pmatrix}.}
\end{align}
Likewise, let $\mathbf{h}$ be the length $3N$ vector defined by
\begin{align}
    \mathbf{h} = \begin{pmatrix}h_0(\cT_1) \cdots h_0(\cT_N) & h(\cT_1) & \cdots & h(\cT_N) & \overline{h}(\cT_1) & \cdots & \overline{h}(\cT_N)\end{pmatrix}^{\mathrm{T}}.
\end{align}
Then the discretized Radon transform, which we denote by $\Delta\mathbf{L}$, is given by the matrix equation
\begin{align}
    \Delta\mathbf{L} = \mathbf{W}\mathbf{h}.\label{eq:discrete_radon}
\end{align}

The discrete inverse Radon transform is then just the inverse problem to the system~\eqref{eq:discrete_radon}. However, the inverse problem is complicated by the fact that the Radon transform is neither surjective nor injective (recall that the forward transform has a non-trivial kernel that corresponds to the gauge degrees of freedom). 

To make the inverse well-posed, we must supplement it with a set of gauge constraints to pick out a unique inverse (when it exists) in the continuum case. In the discrete analog, we expect the kernel to be visible through the presence of zero eigenvalues in the singular value spectrum of $\mathbf{W}$. Due to discretization and numerical errors, the system~\eqref{eq:discrete_radon} will either be singular or extremely ill-conditioned. As in the continuum case, we will supplement the system~\eqref{eq:discrete_radon} with a set of discretized gauge constraints to make the problem well-posed.

For our implementation, we discretize the holomorphic gauge conditions given by equations~\eqref{eq:hologauge1} and~\eqref{eq:hologauge2}. The holomorphic gauge constraints are simple first order partial differential equations. They can be realized discretely by implementing $\partial_x$ and $\partial_y$ are finite-difference operators. We will use a simple three-point stencil for the finite-difference operators, although higher-order variations can be used for increased accuracy. For example, the $x$ partial of a function $f$ at tile $i$ will be approximated by
\begin{align}
    \partial_xf(\cT_i) \approx \frac{f(\cT_{R(i)})-f(\cT_{L(i)})}{2}, 
\end{align}
where $L(i)$ and $R(i)$ denote the indices of the tiles to the immediate left and right of tile $i$.\footnote{This approximation will work for all non-boundary tiles. The tiles on the boundaries will have to use forwards or backwards finite-differences instead. For example, for a tile on the left boundary, we get
\begin{align}
    \partial_xf(\cT_i) \approx \frac{-3f(\cT_i) + 4f(\cT_{R(i)}) - f(\cT_{R(R(i))})}{2}.
\end{align}
} The $y$-partials are analogous. In this way, we can write discretize the partial derivatives as matrices $\tilde{\Delta}_{x}$ and $\tilde{\Delta}_{y}$.\footnote{The entries of $\tilde{\Delta}_x$ are given by 
\begin{align}
(\tilde{\Delta}_x)_{ij} = \frac{\delta_{L(i),j} - \delta_{R(i),j}}{2} 
\end{align}
if tile $i$ is a non-boundary tile, with the appropriate modifications for the boundaries.
}
Since the vector $\mathbf{h}$ effectively contains $3$ copies stacked on top of each other, we define the operators  
\begin{align}
    \Delta_x = \left(\tilde{\Delta}_x \mid \tilde{\Delta}_x \mid \tilde{\Delta}_x\right),\label{eqn:discrete_derivative}\\
    \Delta_y = \left(\tilde{\Delta}_y \mid \tilde{\Delta}_y \mid \tilde{\Delta}_y\right).\nonumber
\end{align}
The discretized holomorphic gauge conditions can then be written as
\begin{align}
    \mathbf{C}\mathbf{h} = (\Delta_x + i\Delta_y + 2\Lambda)\mathbf{h},\label{eq:discrete_gauge_constraints}
\end{align}
where $\Lambda$ is the matrix defined by
\begin{align}
    \Lambda = \left(\tilde{\Lambda}\mid \tilde{\Lambda} \mid \tilde{\Lambda}\right),
\end{align}
where $\tilde{\Lambda}$ has entries given by
\begin{align}
    \tilde{\Lambda}_{ij} = \partial\lambda(\cT_i)\,\delta_{ij}.
\end{align}

\section{Accuracy of the Numerical Reconstruction}
\label{app:benchmark}

In the absence of an analytic reconstruction formula in the continuum case, we need to validate the numerical reconstruction so as to provide confidence that it performs the inverse transform correctly. To do so, we will first benchmark the numerical reconstruction using various test cases obtained by instantiating known rank-2 symmetric tensor fields in the bulk in the holomorhpic gauge. From equation~\eqref{eq:holo_gauge_solution}, we can see that a bulk solution in the holomorphic gauge can be readily generated by prescribing the boundary values.

As a benchmarking test, we therefore test the reconstruction algorithm as follows:
\begin{enumerate}
\item We first pick an arbitrary function $h|_{\partial\bbD}:\partial\bbD \rightarrow \bbC$ on the unit circle. We also pick another arbitrary scalar function $h_0:\bbD \rightarrow \bbC$.

\item Using equation~\eqref{eq:holo_gauge_solution}, we extend the function $h|_{\partial \bbC}$ into the bulk to get a holomorphic function $h:\bbD\rightarrow \bbD$.

\item From the functions $h_0$ and $h$, we define the bulk tensor
\begin{align}
    f = h_0 + he^{2i\theta} + \overline{h}e^{-2i\theta},
\end{align}
where $\overline{h}$ denotes the complex conjugate of $h$. The $f$ defined in this way is a real-valued continuum bulk $2$-tensor with components fixed in the holomorphic gauge.

\item We run the forward numerical Radon transform to generate boundary data $\mathbf{b}$. From $\mathbf{b}$, we then run the discretized inverse Radon transform with holomorphic gauge constraints to numerically recover a discretized bulk solution $\mathbf{h}$.

\item Finally, we compare the values of $h$, evaluated at the centroids of the discretized tiling, with the reconstructed value $\mathbf{h}$. Denoting the exact solution at the centroids as $\mathbf{h}_*$, we can evaluate the relative error
\begin{align}
    \cE_{\mathrm{bulk}} = \frac{\|\mathbf{h} - \mathbf{h}_*\|}{\|\mathbf{h}_*\|}
\end{align}
to get a sense of the reconstruction quality.
\end{enumerate}

In the absence of known analytic results, and before we move onto with physically relevant examples, we can run the above test with several choices of boundary functions $h|_{\partial \bbD}$ as a proof of concept that the numerical inverse Radon transform is performing an adequate reconstruction. 

Below, we show some sample test cases for the numerical inverse transform. All of the transforms shown below are performed with $1000$ bulk tiles and $100$ boundary sites (for a total of $4950$ geodesics), following the procedure outlined above.

The reconstructions here show good agreement with the original bulk data. On a qualitative level, the plots of the bulk metric profile are essentially indistinguishable between the original and the reconstruction. The bulk relative errors across various test cases range from $0.01$ to $0.1$, indicating quantitatively successful reconstructions in general. The boundary relative errors are typically an order of magnitude smaller than the bulk relative errors, ranging from $0.001$ to $0.01$. The magnitudes of the boundary relative errors for these known test cases can serve as an estimate for the general magnitude of numerical errors present in the algorithm. See Figures~\ref{fig:benchmark1}-\ref{fig:benchmark4}.\footnote{Note that the errors between the original and the reconstruction would be better illustrated as difference plots. Due to the unfortunately circumstances around the world during the writing of this paper, some of the reconstruction data is currently quarantined away from one of the authors. Difference plots will be added in a later version.}

It should be noted that the relative errors shown here can actually be slightly misleading. Most sources of error in the reconstruction arise due to large fluctuating values of the boundary tiles. Tiles at or near the boundary are generally underconstrained due to the relatively small number of geodesics which pass through any given boundary region. This allows the boundary tiles to take on arbitrary values in order to minimize the relative boundary error as the algorithm is designed to do, although this comes at the cost of bulk accuracy. We can see that if we exclude the the values of the boundary tiles from the calculation of relative error, that the relative error of the reconstruction is generally an order of magnitude smaller. This suggests that the bulk reconstruction performs very well deep into the bulk, with larger errors towards the boundary. This is important to keep in mind, since it suggests that the qualitative bulk picture provided by the numerical inverse should be generally trustworthy.

\begin{figure}
    \centering
    \includegraphics[width=0.6\textwidth]{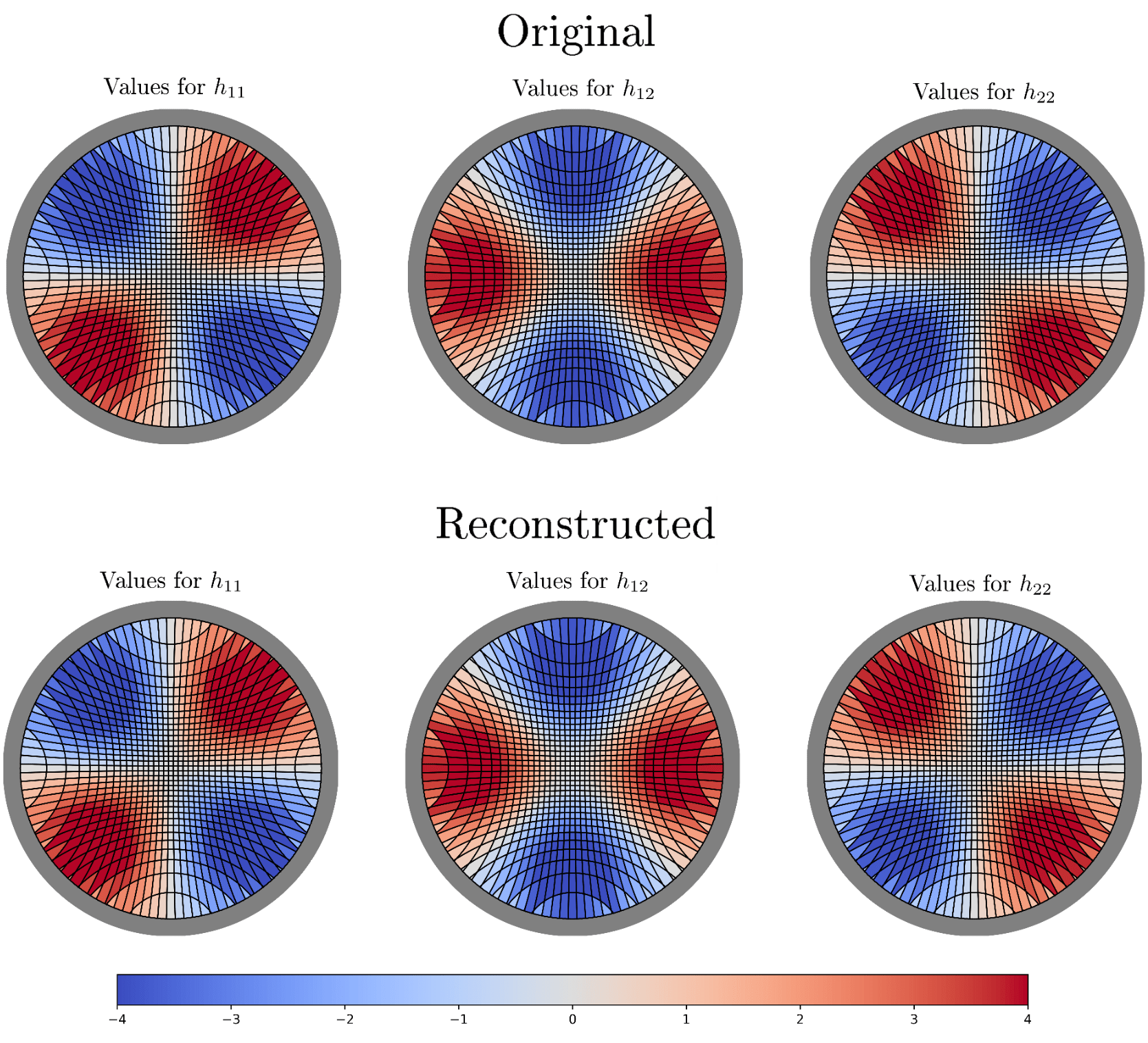}
    \caption{\textbf{Top:} Original bulk data generated with $h_0(x,y) = xy$ and $h|_{\bbC}(\theta)=\sin(2\theta)$. \textbf{Bottom:} The bulk data reconstructed after running the forward Radon transform to get boundary data. Visually, the two sets of data are identical. The relative error between the two are $\cE_{\mathrm{bulk}} \approx 0.0227$ (relative error $0.004643$ without boundary tiles).}
    \label{fig:benchmark1}
\end{figure}

\begin{figure}
    \centering
    \includegraphics[width=0.6\textwidth]{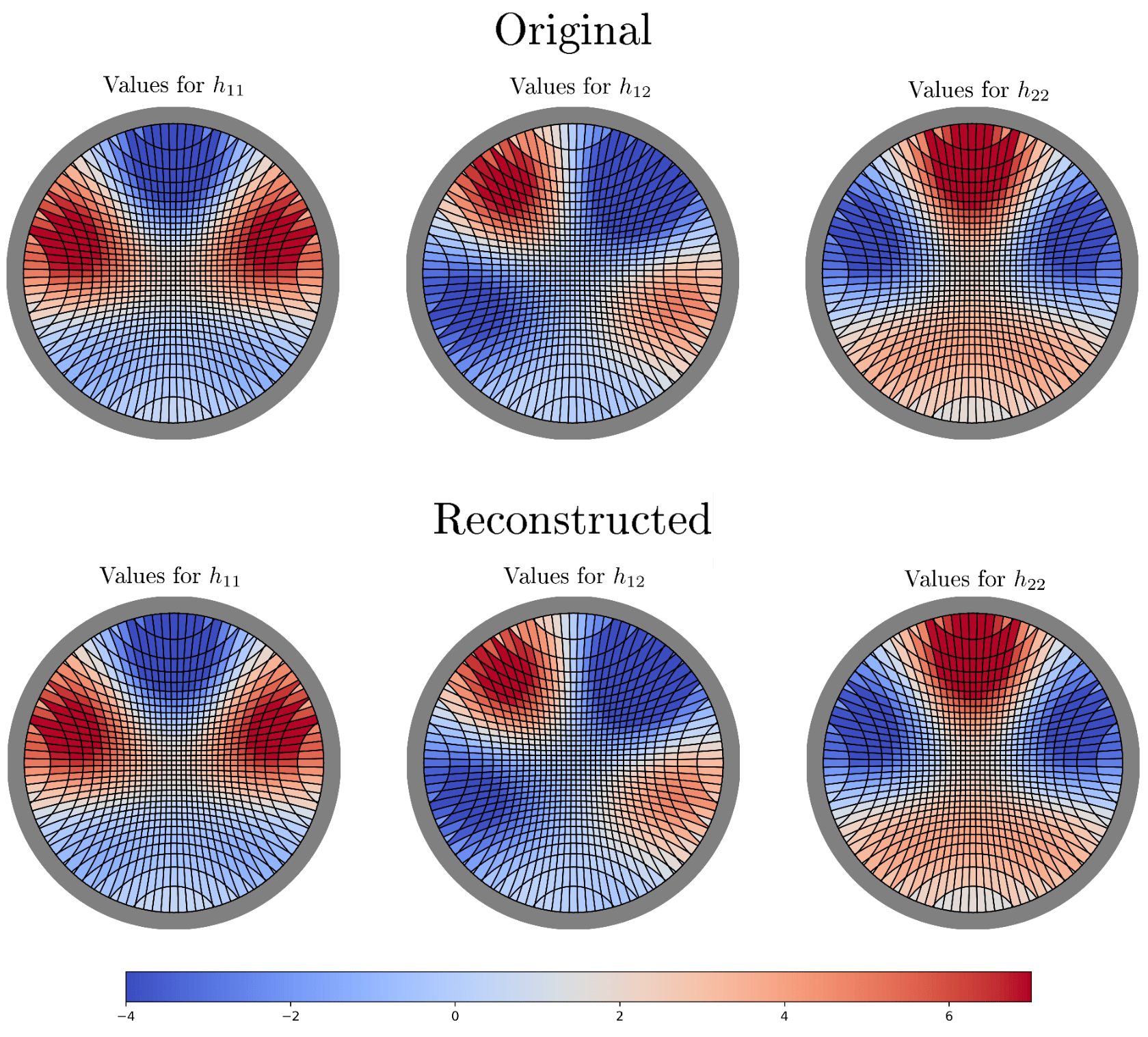}
    \caption{\textbf{Top:} Original bulk data generated with $h_0(x,y) = 2e^{-(x^2+y^2)}$ and $h|_{\bbC}(\theta)=\cos(2\theta) + \sin(3\theta)$. \textbf{Bottom:} The bulk data reconstructed after running the forward Radon transform to get boundary data. Visually, the two sets of data are identical. The relative error between the two are $\cE_{\mathrm{bulk}} \approx 0.0273$ (relative error $0.005516$ without boundary tiles).}
    \label{fig:benchmark2}
\end{figure}

\begin{figure}
    \centering
    \includegraphics[width=0.6\textwidth]{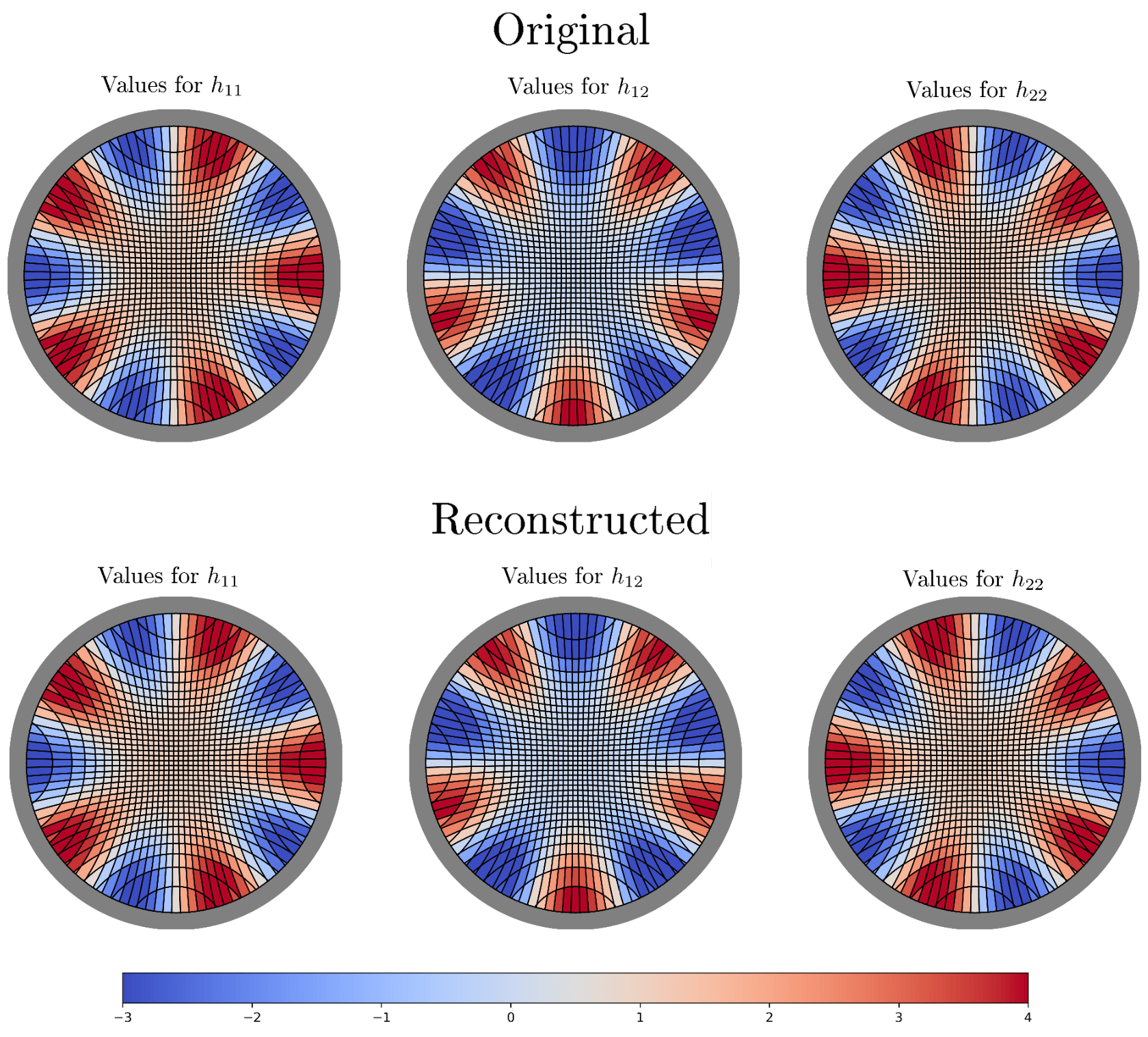}
    \caption{\textbf{Top:} Original bulk data generated with $h_0(x,y) = 1/(1+x^2+y^2)$ and $h|_{\bbC}(\theta)=\cos(5\theta)$. \textbf{Bottom:} The bulk data reconstructed after running the forward Radon transform to get boundary data. Visually, the two sets of data are identical. The relative error between the two are $\cE_{\mathrm{bulk}} \approx 0.0885$ (relative error $0.01357$ without boundary tiles).}
    \label{fig:benchmark3}
\end{figure}

\begin{figure}
    \centering
    \includegraphics[width=0.6\textwidth]{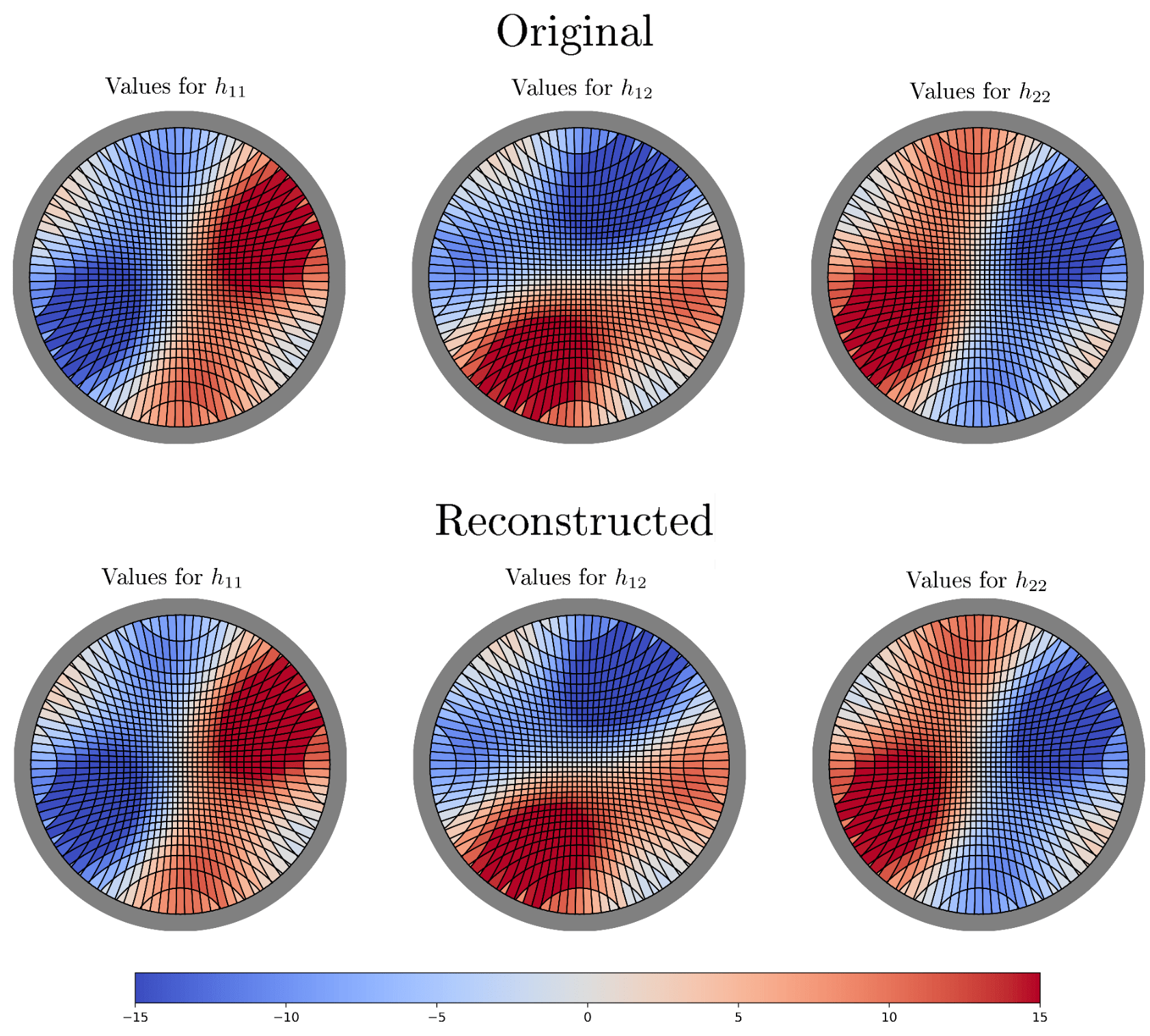}
    \caption{\textbf{Top:} Original bulk data generated with $h_0(x,y) = x^2+y^2$ and $h|_{\bbC}(\theta)=2\cos(\theta)+3\sin(3\theta)$. \textbf{Bottom:} The bulk data reconstructed after running the forward Radon transform to get boundary data. Visually, the two sets of data are identical. The relative error between the two are $\cE_{\mathrm{bulk}} \approx 0.01886$ (relative error $0.003709$ without boundary tiles).}
    \label{fig:benchmark4}
\end{figure}

\subsection{Constrained Optimization}
With the discrete Radon transform~\eqref{eq:discrete_radon} and the discretized holomorphic constraints~\eqref{eq:discrete_gauge_constraints}, we can solve this linear system for $\mathbf{h}$ as a constrained optimization problem:


\begin{align}
   \min_{\mathbf{h}}\,\|\mathbf{W}\mathbf{h}-\mathbf{b}\|,\label{eqn:conslstsq}\\
    \text{subject to}\ \mathbf{C}\mathbf{h}=\mathbf{0},\nonumber
\end{align}
We look for a best-fit solution $\mathbf{h}_*$ that solves the above system. In general, we do not expect there to exist a solution $\mathbf{h}_*$ such that $\mathbf{W}\mathbf{h}_*-\mathbf{b}=\mathbf{0}$, due to either numerical/discretization errors or the boundary data being non-geometric.

Because the objective function is linear, this problem has a unique global solution that can be obtained using standard constrained least squares. We briefly review the method\cite{} below. 
 
\begin{theorem}
Consider the constrained least square problem~\eqref{eqn:conslstsq}. Assuming the stacked matrix

\begin{equation}\label{eq:stacked_matrix}
\begin{pmatrix}
\mathbf{W}\\
\mathbf{C}
\end{pmatrix}
\end{equation}

is left-invertible and $\mathbf{C}$ is right-invertible, a vector $\mathbf{h}_*$ uniquely solves the constrained least square problem~\eqref{eqn:conslstsq} if and only if there exists some $\mathbf{g}$ such that
\begin{equation}
    \begin{pmatrix}
    \mathbf{W}^\dagger\mathbf{W} & \mathbf{C}^\dagger\\
    \mathbf{C} & \mathbf{0}
    \end{pmatrix}
    \begin{pmatrix}
    \mathbf{h}_*\\
    \mathbf{g}
    \end{pmatrix}=
    \begin{pmatrix}
    \mathbf{W}^\dagger\mathbf{b}\\
    \mathbf{0}
    \end{pmatrix}.
    \label{eqn:lstsqr}
\end{equation}
\end{theorem} 
 
\begin{proof}
Suppose that $(\mathbf{h}_*, \mathbf{g})$ satisfies~\eqref{eqn:lstsqr}. Clearly $\mathbf{h}_*$ satisfies the constraint $\mathbf{C}\mathbf{h}_* = \mathbf{0}$. Then for any $\mathbf{h}$ that satisfies the constraint $\mathbf{C}\mathbf{h}=\mathbf{0}$, we have
\begin{align}
     \|\mathbf{Wh}-\mathbf{b}\|^2 &= \|\mathbf{W(h-h_*)}+\mathbf{W}\mathbf{h}_*-\mathbf{b}\|^2\\
     &=\|\mathbf{W(h-h_*)}\|^2 +\|\mathbf{W}\mathbf{h}_*-\mathbf{b}\|^2 +2(\mathbf{h}-\mathbf{h}_*)^\dagger\mathbf{W}^\dagger(\mathbf{Wh_*}-\mathbf{b})\\
     &=||\mathbf{W(h-h_*)}||^2 +||\mathbf{Wh_*}-\mathbf{b}||^2 + 2(\mathbf{h}-\mathbf{h_*})^\dagger\mathbf{C}^\dagger\mathbf{g}\\
     &=||\mathbf{W(h-h_*)}||^2 +||\mathbf{Wh_*}-\mathbf{b}||^2\\
     &\geq ||\mathbf{Wh_*}-\mathbf{b}||^2,
\end{align}
where in the second line we used the definition of the norm $\|\mathbf{v}\|^2 = \mathbf{v}^\dagger\mathbf{v}$, in the third line we used the fact that 
\begin{align}
   \mathbf{W}^\dagger\mathbf{Wh_*} +\mathbf{C}^\dagger\mathbf{g}=\mathbf{W}^\dagger\mathbf{b},
\end{align}
and in the fourth line we used the fact that
\begin{align}
    \mathbf{C(h-h_*)}=\mathbf{Ch}-\mathbf{Ch_*}=\mathbf{0}-\mathbf{0}=\mathbf{0}. 
\end{align}
Therefore, $\mathbf{h_*}$ is an optimal solution. Furthermore, the optimal constrained solution is obtained if and only if
 \begin{equation}
     \begin{pmatrix}
     \mathbf{W}\\
     \mathbf{C}
     \end{pmatrix}
     (\mathbf{h-h_*})=\mathbf{0}.
 \end{equation}
By assumption the stacked matrix is left-invertible, therefore the minimum $\mathbf{h}=\mathbf{h_*}$ is also unique.

In the reverse direction, it suffices to show the matrix in~\eqref{eqn:lstsqr} is invertible. Suppose on the contrary the matrix is non-invertible. Then there must exist a vector $(\mathbf{h}, \mathbf{g}) \ne \mathbf{0}$ such that
\begin{equation}\label{eq:lstsq1}
        \begin{pmatrix}
    \mathbf{W}^\dagger\mathbf{W} & \mathbf{C}^\dagger\\
    \mathbf{C} & \mathbf{0}
    \end{pmatrix}
    \begin{pmatrix}
    \mathbf{h}\\
    \mathbf{g}
    \end{pmatrix}=
    \begin{pmatrix}
    \mathbf{0}\\
    \mathbf{0}
    \end{pmatrix}.
\end{equation}
Left multiplying both sides by $(\mathbf{h},\mathbf{g})^\dagger$, we have 
 \begin{equation}
     \begin{pmatrix}\mathbf{h}^\dagger& \mathbf{0}^\dagger\end{pmatrix}\begin{pmatrix}\mathbf{W}^\mathrm{T}\mathbf{W} & \mathbf{C}^\dagger\\
     \mathbf{C} & \mathbf{0} \end{pmatrix} \begin{pmatrix} \mathbf{h} \\ \mathbf{g}\end{pmatrix} = \mathbf{h}^\dagger\mathbf{W}^\dagger\mathbf{W}\mathbf{h} + \mathbf{h}^\dagger\mathbf{C}^\dagger\mathbf{g} = 0.
 \end{equation}
Noting that $\mathbf{C}\mathbf{h}^\dagger=\mathbf{0}$, we get $\|\mathbf{W}\mathbf{h}\|^2 = 0$, so that $\mathbf{W}\mathbf{h} = \mathbf{0}$. Since the stacked matrix~\eqref{eq:stacked_matrix} is injective, it follows that we must have $\mathbf{h} = \mathbf{0}$. The system~\eqref{eq:lstsq1} then reduces to
\begin{equation}
    \mathbf{C}^\dagger\mathbf{g} = \mathbf{0}.
\end{equation}
We then conclude that $\mathbf{g} = \mathbf{0}$ since $\mathbf{C}$ is right-invertible by assumption. This is in contradiction with the assumption that $(\mathbf{h},\mathbf{g})\neq \mathbf{0}$. Hence the matrix in equation~\eqref{eqn:lstsqr} must be invertible, and a solution for $\mathbf{g}$ must exist.
\end{proof}

\subsection{Interpolation and Regularization}
  
In principle, the optimal solution $\mathbf{h}_*$ can be obtained through straightforward matrix inversion of equation~\eqref{eqn:lstsqr}. Suppose $\mathbf{h}$ is a column vector with $3N$ entries, and the constraint matrix is $M\times 3N$, then any well-known polynomial algorithm for matrix inversion is of $\cO\left((M+3N)^3\right)$. The inversion of system~\eqref{eqn:lstsqr} is slightly complicated by the fact that the matrix $\mathbf{W}$ is generally expected to be extremely ill-conditioned. We can get around this by regularizing the system.

To make the the least squares problem~\eqref{eqn:lstsqr} better conditioned, and to smooth out small scale fluctuations in the discretized reconstruction, we employ a derivative type Tikhonov regularization. We replace the matrix $\mathbf{W}^\dagger \mathbf{W}$ in~\eqref{eqn:lstsqr} by
\begin{align}
    \mathbf{W}^\dagger \mathbf{W} + \gamma \left(\Delta_x^\dagger\Delta_x + \Delta_y^\dagger\Delta_y\right),
\end{align}
where $\gamma>0$ is a regularization parameter which controls the strength of the regularization, and where $\Delta_x$ and $\Delta_y$ are the discretized partial differential operators defined in~\eqref{eqn:discrete_derivative}. This replacement effectively changes the least squares problem in~\eqref{eqn:conslstsq} to
\begin{align}
   \min_{\mathbf{h}}\,\left(\|\mathbf{W}\mathbf{h}-\mathbf{b}\|^2+\gamma\|\Delta_x\mathbf{h}\|^2 + \gamma\|\Delta_y\mathbf{h}\|^2\right),\\
    \text{subject to}\ \mathbf{C}\mathbf{h}=\mathbf{0},\nonumber
\end{align}
which both regularizes the system so that the smallest singular values of $\mathbf{W}$ are of order $~\sqrt{\gamma}$ and also takes into account the strength of fluctuations in the resulting solution. Since we expect small scale fluctuations to be mostly due to bulk discretization errors, this choice of regulator serves as a reasonable filter. In this note, all of our reconstructions employ regularization with $\gamma= 10^{-8}$.

In the case of fixed data but variable number of bulk tiles, the reconstruction can also become ill-conditioned when the number of bulk degrees of freedom exceed the number of boundary constraints in the form of geodesic lengths. Roughly speaking, because we have $O(3N)$ number of bulk degrees of freedom, one for each tensor component at a particular tile. Suppose there are $K$ sites on the boundary, then the number of linear equations from geodesic lengths is of order $O(K^2)$. Hence the reconstruction can become ill-conditioned when $3N>K^2$. While it is possible to decrease the number of bulk tiles, or increase the number of boundary sites, both come at a cost depending on our requirements for reconstruction. As an alternative, we can also interpolate between geodesic data by adding virtual sites and the lengths of geodesics for the extended set of sites. 

In the current implementation, we add the virtual sites in between the original lattice sites such that the new lattice scale is half of the original. Let us label all sites sequentially along the counter clockwise direction on the boundary circle from $1$ through $2K$ such that $2K+1 \equiv 1$. To generate the geodesic lengths between a virtual site $j$ and an original site $i$, we average the lengths of two geodesics that are anchored at $(i,j+1)$ and $(i,j-1)$. For a geodesic that ends on two virtual sites, $i,j$. We take the 4-point average of the original geodesic lengths for the ones anchored at $(i,j+1), (i,j-1), (i+1,j),$ and $(i-1,j)$. 

\bibliographystyle{unsrt}
\bibliography{radon}

\end{document}